\documentclass[aps,pra,twocolumn,nofootinbib, superscriptaddress,10pt]{revtex4-2}
\newif\ifapp
\newcommand{\lsp}{\hspace{0.1em}}

\usepackage{float}
\makeatletter
\let\newfloat\newfloat@ltx
\makeatother

\usepackage{amsmath,amsfonts,amssymb,amsthm,bm}
\usepackage{mathrsfs}
\usepackage{enumitem}
\usepackage{yhmath}
\usepackage{url}
\usepackage{bbm}
\usepackage{graphicx}
\usepackage{subfigure}
\usepackage{lmodern}
\usepackage[dvipsnames]{xcolor}
\usepackage[colorlinks=true,citecolor=blue,linkcolor=blue]{hyperref}
\usepackage{csquotes}
\MakeOuterQuote{"}
\usepackage{mathtools,mleftright}
\mleftright
\newtheorem{theorem}{Theorem}

\newtheorem{lemma}{Lemma}

\newtheorem{proposition}{Proposition}

\newtheorem{corollary}{Corollary}
\theoremstyle{plain}

\newcommand{\Null}{\operatorname{Null}}
\newcommand{\rank}{\operatorname{rank}}
\newcommand{\tr}{\operatorname{tr}}
\newcommand{\supp}{\operatorname{supp}}
\newcommand{\range}{\operatorname{range}}

\newcommand{\imply}{\mathrel{\Rightarrow}}

\newcommand{\rms}{\mathrm{s}}

\newcommand{\caH}{\mathcal{H}}
\newcommand{\caM}{\mathcal{M}}
\newcommand{\caP}{\mathcal{P}}

\newcommand{\caV}{\mathcal{V}}

\newcommand{\scrA}{\mathscr{A}}
\newcommand{\scrB}{\mathscr{B}}
\newcommand{\scrC}{\mathscr{C}}
\newcommand{\scrE}{\mathscr{E}}

\newcommand{\scrM}{\mathscr{M}}

\newcommand{\bbS}{\mathbb{S}}

\newcommand{\bPi}{\bar{\Pi}}

\newcommand{\bmp}{\bm{p}}
\newcommand{\bmq}{\bm{q}}

\newcommand{\bbone}{\mathbbm{1}}

\def\eqref#1{\textup{(\ref{#1})}}
\newcommand{\eref}[1]{Eq.~\textup{(\ref{#1})}}
\newcommand{\eqsref}[2]{Eqs.~(\ref{#1}) and (\ref{#2})}

\newcommand{\Eref}[1]{Equation~\textup{(\ref{#1})}}
\newcommand{\Eqsref}[2]{Equations~(\ref{#1}) and (\ref{#2})}

\newcommand{\lref}[1]{Lemma~\ref{#1}}
\newcommand{\lsref}[1]{Lemmas~\ref{#1}}

\newcommand{\thref}[1]{Theorem~\ref{#1}}

\newcommand{\Thref}[1]{Theorem~\ref{#1}}

\newcommand{\pref}[1]{Proposition~\ref{#1}}
\newcommand{\psref}[1]{Propositions~\ref{#1}}
\newcommand{\Pref}[1]{Proposition~\ref{#1}}
\newcommand{\Psref}[1]{Propositions~\ref{#1}}

\newcommand{\coref}[1]{Corollary~\ref{#1}}

\newcommand{\fref}[1]{Fig.~\ref{#1}}

\newcommand{\aref}[1]{Appendix~\ref{#1}}
\newcommand{\asref}[1]{Appendices~\ref{#1}}

\def\<{\langle}  
\def\>{\rangle}  

\newcommand{\rcite}[1]{Ref.~\cite{#1}}

\begin{document}
	\title{Optimal Quantum Measurements with respect to the Fidelity}
	\author{Datong Chen}
	\author{Huangjun Zhu}
	\email{zhuhuangjun@fudan.edu.cn}
	\affiliation{State Key Laboratory of Surface Physics, Department of Physics, and Center for Field Theory and Particle Physics, Fudan University, Shanghai 200433, China}
	
	\affiliation{Institute for Nanoelectronic Devices and Quantum Computing, Fudan University, Shanghai 200433, China}
	
	\affiliation{Shanghai Research Center for Quantum Sciences, Shanghai 201315, China}

	\date{\today}
	
\begin{abstract}
Fidelity is the standard measure for quantifying the similarity between two quantum states. It is equal to the square of the minimum Bhattacharyya coefficient between the probability distributions induced by quantum measurements on the two states. Though established  for over thirty years, the structure of fidelity-optimal quantum measurements  remains unclear  when the two density operators are singular (not invertible). Here we address this gap, with a focus on minimal optimal measurements, which admit no nontrivial coarse graining that is still optimal. We show that there exists either a unique minimal optimal measurement or infinitely many inequivalent choices. Moreover, the first case holds if and only if the two density operators satisfy a weak commutativity condition. In addition, we provide a complete characterization of all minimal optimal measurements when one state is pure, leveraging geometric insights from the Bloch-sphere representation. The connections with quantum incompatibility, operator pencils, and geometric means are highlighted.  
\end{abstract}

	\maketitle

\emph{Introduction}---Nonorthogonal quantum states cannot be perfectly distinguished; this fundamental fact is tied to the no-cloning theorem and has profound implications for quantum information processing. Constructing optimal quantum measurements for distinguishing such states is instrumental to many practical applications, including quantum communication, quantum cryptography, quantum metrology, and quantum computation. Not surprisingly, the optimal measurements depend on the specific measure employed to quantify the distinguishability or similarity between quantum states. Fidelity is one of the most popular similarity measures because of its simplicity and appealing properties \cite{Fuchs1996Measures,Fuchs1999Measures,Nielsen2010Quantum,Wilde2017Information, Liang2019Fidelity}. Notably, the fidelity between two pure states is equal to the transition probability, and the fidelity between two mixed states is equal to the maximum fidelity between their purifications thanks to Uhlmann and Jozsa \cite{Uhlmann1976Fidelity,Richard1994Fidelity}. Moreover,  the fidelity is equal to the square of the minimum Bhattacharyya coefficient \cite{Bhattacharyya1943Fidelity,Wootters1981Fidelity} between probability distributions induced by quantum measurements on the two states thanks to Fuchs and  Caves \cite{Fuchs1995distinguishability, Fuchs1996Measures}. Although these facts have been known for more than three decades,  the structure of optimal measurements is still quite elusive  when the two density operators are singular.

In this work, we clarify the structure of optimal quantum measurements for distinguishing two general quantum states with respect to the fidelity.  To this end, we introduce the concept of minimal optimal measurements as those optimal measurements that admit no nontrivial coarse graining preserving optimality. By virtue of connections with operator pencils \cite{Ikramov1993Pencils,Golub2013matrix} and geometric means \cite{Lawson2001GM, Bhatia2007GeoMean},
we demonstrate a sharp dichotomy: there exists either a unique minimal optimal measurement or infinitely many inequivalent choices. Notably, the first case holds if and only if (iff) the two density operators satisfy a weak commutativity condition (assuming that their sum is nonsingular): the projectors onto their supports commute, but the density operators  do not have to commute. In addition, we show that there exists a unique minimal optimal measurement that commutes with the projector onto the support of one density operator. Furthermore, we completely characterize all minimal optimal measurements when one quantum state is pure and provide geometric intuition via Bloch-sphere representation in the case of two pure states. In the course of study, we derive a number of results on operator pencils and geometric means, which are of independent interest.

\emph{Preliminaries}---Let $\caH$ be a $d$-dimensional Hilbert space, and denote by $\bbone$ the identity operator on $\caH$. Given a linear operator $A$ on $\caH$, denote by $\Null(A)$ the null space of $A$ and by  $\supp(A)$ the support of $A$, that is,  the orthogonal complement of $\Null(A)$. Quantum states on $\caH$ are usually described by density operators. Recall that a density operator  $\rho$ on $\caH$ is a positive  (semidefinite) operator of unit trace. Quantum measurements on $\caH$ can be characterized by positive operator-valued measures (POVMs) \cite{Nielsen2010Quantum,Neumann2018POVM} since we are not interested in post-measurement quantum states in this work. Recall that a POVM  $\scrE=\{E_m\}_m$  is composed of a collection of positive operators that sum up to the identity operator~$\bbone$,
that is, $E_m\geq 0$ for all $m$ and $\sum_m E_m=\bbone$.  A POVM is a projector-valued measure (PVM) if all its POVM elements are mutually orthogonal projectors.
If we perform  the POVM $\scrE=\{E_m\}_m$ on $\rho$, then the probability of obtaining outcome $m$ reads $p_m=\tr\left(E_m\rho\right)$. In this way, the POVM $\scrE$ induces a map from quantum states on $\caH$ to probability distributions. The probability distribution associated with $\rho$ and $\scrE$ is denoted by
\begin{equation}\label{Eq:MeasurementChannel}
\Lambda_\scrE(\rho):=\{\tr\left(E_m\rho\right)\}_m. 
\end{equation}

Next, we introduce an order relation on POVMs from the perspective of data processing \cite{Martens1990Nonideal,Kuramochi2015Minimal, Zhu2022POVM}. Given two POVMs $\scrA=\{A_j\}_j$ and $\scrB=\{B_k\}_k$  on $\caH$, the POVM $\scrA$ is a \emph{coarse graining} of $\scrB$,  if each POVM element of $\scrA$ can be expressed as follows:
\begin{equation}\label{eq:DefCoarse}
    A_j = \sum_k S_{jk} B_k,
\end{equation}
where $S$ is a stochastic matrix. Accordingly, $\scrB$ is a \emph{refinement} of $\scrA$. 
Two POVMs are \emph{equivalent} if they are coarse grainings of each other. 
A coarse graining or refinement of  $\scrA$ is trivial if it is equivalent to $\scrA$ and nontrivial otherwise. 
A POVM is \emph{simple} if no POVM element is proportional to another POVM element.
Every POVM is equivalent to a simple POVM; in addition, two simple POVMs are equivalent iff they are identical up to relabeling (or permutation of POVM elements). 
In this paper, we are mainly interested in simple POVMs.

Next, $\scrA$ commutes with $\scrB$ 
if all POVM elements in $\scrA$ commute with all POVM elements in $\scrB$. Notably, any coarse graining or refinement of a PVM commutes with the PVM. By contrast, $\scrA$ and $\scrB$  are \emph{compatible} or \emph{jointly measurable} if they admit a common refinement and \emph{incompatible}  otherwise \cite{Busch1986Compatible,Guhne2023Compatible,Zhu2022POVM}.  Commuting POVMs are compatible, but the converse is not guaranteed in general. Nevertheless, two PVMs are compatible iff they  commute~\cite{Lahti2003coexistence,Guhne2023Compatible}.

\emph{Optimal measurements with respect to the fidelity}---Let  $\bmp=\{p_m\}_m$ and $\bmq=\{q_m\}_m$ be two probability distributions indexed by a same set. The classical fidelity (also known as the square of the Bhattacharyya coefficient) between  $\bmp$ and $\bmq$ reads  \cite{Bhattacharyya1943Fidelity,Wootters1981Fidelity}
 \begin{equation}\label{Eq:Bhattacharyya1943Fidelity}
F(\bmp,\bmq):=\left(\sum_m \sqrt{p_mq_m}\right)^2,
\end{equation}
which is symmetric in $\bmp$ and $\bmq$ by definition. 
As a generalization, the fidelity between two density operators $\rho$ and $\sigma$ on $\caH$ reads \cite{Uhlmann1976Fidelity,Richard1994Fidelity}
\begin{equation}\label{Eq:QuantumFidelity}
F(\rho, \sigma) := \left(\tr\sqrt{\sqrt{\rho}\lsp\sigma\sqrt{\rho}}\lsp\right)^2=\left(\tr\sqrt{\sqrt{\sigma}\rho\sqrt{\sigma}}\lsp\right)^2,
\end{equation}
where the second equality holds because $\sqrt{\rho}\sqrt{\sigma}$ and  $\sqrt{\sigma}\sqrt{\rho}$ share the same singular values, including  multiplicities. So $F(\rho, \sigma)$ is symmetric in $\rho$ and $\sigma$, which is analogous to the classical counterpart.

To understand the connection between the quantum  and classical fidelities, given a POVM $\scrE$ on $\caH$, denote by $F_\scrE(\rho,\sigma)$ 
 the classical fidelity between 
$\Lambda_\scrE(\rho)$ and $\Lambda_\scrE(\sigma)$ as shown in \fref{Fig:Foptimal}. Then $F_\scrE(\rho,\sigma)\geq F(\rho, \sigma)$ \cite{Fuchs1995distinguishability, Fuchs1996Measures,Nielsen2010Quantum},
and this lower bound is tight. 
The POVM $\scrE$ is \emph{F-optimal} for distinguishing $\rho$ and $\sigma$ if this bound is saturated, that is, $F_\scrE(\rho,\sigma)=F(\rho, \sigma)$.  In addition, the induced classical fidelity is monotonic under data processing as shown in the following proposition, which is  a corollary of Theorem~9 in~\rcite{Van2014FCoarse}. See  \asref{app:FoptimalProofs} and \ref{app:PureStateProof} for proofs of results presented in the main text. 
\begin{proposition}\label{pro:CoarseIneq}
Suppose $\rho, \sigma$ are two quantum states on $\caH$ and $\scrA,\scrB$ are two POVMs on $\caH$, where $\scrA$ is a coarse graining of $\scrB$. Then
\begin{equation}\label{eq:CoarseIneq}
F_\scrA(\rho,\sigma)\ge F_\scrB(\rho,\sigma).
\end{equation}
\end{proposition}
As a simple corollary of \pref{pro:CoarseIneq}, equivalent POVMs induce the same classical fidelity as expected. In addition, any refinement of an  F-optimal POVM is still  F-optimal. In view of this fact, an F-optimal POVM $\scrE$ is   \emph{minimal} if it is simple and any nontrivial coarse graining is not F-optimal; in other words, any coarse graining of $\scrE$ that preserves optimality is equivalent to~$\scrE$. Such an optimal POVM extracts the least information and causes the least disturbance and is thus of special interest. In addition, the probability distributions induced by $\scrE$ only depend on its restriction on $\supp(\rho+\sigma)$. In the rest of this paper, unless otherwise stated,  we assume that $\rho$ and $\sigma$ are two distinct quantum states on $\caH$,
$\rho+\sigma$ is nonsingular, and $\scrE=\{E_m\}_m$ is a simple POVM on~$\caH$.

\begin{figure}
	\centering
	\includegraphics[width=0.42\textwidth]{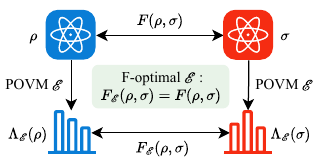}
	\caption{Schematic diagram of the fidelity $F(\rho, \sigma)$ and the classical fidelity $F_\scrE(\rho,\sigma)$ induced by the POVM $\scrE$. The POVM $\scrE$ is F-optimal if $F_\scrE(\rho,\sigma)=F(\rho, \sigma)$. }
	\label{Fig:Foptimal}
\end{figure}

\emph{Insights from operator pencils}---To clarify the conditions for constructing F-optimal POVMs, we need to introduce an additional concept. Two operators $K$ and $L$ on $\caH$ are \emph{parallel} if 
\begin{equation}\label{eq:ParallelDe}
L=0 \quad\text{or}\quad K=\lambda L,
\end{equation} 
where $\lambda$ is a nonnegative number. This condition is important to us because it arises naturally in saturating operator Cauchy-Schwarz inequalities tied to F-optimal POVMs. The following proposition offers a necessary and sufficient condition for a POVM to be F-optimal, which refines a result in \rcite{Fuchs1996Measures}. 
\begin{proposition}\label{prop:FoptConParallel}
Suppose  $U$ is any unitary operator that satisfies $\sqrt{\rho}\sqrt{\sigma}\lsp U = \sqrt{\sqrt{\rho}\lsp\sigma\sqrt{\rho}}$. 
Then $\scrE$ is F-optimal for distinguishing $\rho$ and $\sigma$ iff $\sqrt{\sigma}\sqrt{E_m}$ and $U\sqrt{\rho}\sqrt{E_m}$ are parallel for  each POVM element $E_m$ in $\scrE$, that is,
\begin{equation}\label{eq:ConditionFidelity1}
\sqrt{\rho}\sqrt{E_m} = 0\quad\text{or}\quad  \sqrt{\sigma}\sqrt{E_m}=\kappa_mU\sqrt{\rho}\sqrt{E_m},
\end{equation}
where $\kappa_m$ is a nonnegative number.
\end{proposition}


\Eref{eq:ConditionFidelity1} above may be regarded as a generalized eigenequation. To get a deeper understanding about this condition, it is instructive to introduce the concept of linear operator pencils or pencils for short, which are useful in a number of research areas, including entanglement classification~\cite{Chitambar2010PencilEntangle,Slowik2020PencilEntangle} and quantum learning of many-body systems~\cite{Steffens2014PencilManybody,Kaneko2025PencilManybody}.
Here a \emph{pencil} refers to a pair $(K,L)$ of linear operators on $\caH$ \cite{Golub2013matrix} (see \aref{sec:Pencil} for more details).  A complex number $\lambda$ is an eigenvalue of $(K,L)$ if $K-\lambda L$ is singular. In that case, $\Null(K-\lambda L)$ is called the eigenspace of $(K,L)$ with eigenvalue $\lambda$, and any nonzero vector in this space is called an eigenvector. If $L=\bbone$, then the eigenvalues (eigenspaces) of $(K,L)$ coincide with the counterparts of $K$. If instead $L$ is singular, then $\infty$ is also regarded as an eigenvalue of $(K,L)$, and $\Null(L)$ is regarded as the corresponding eigenspace.  Using the language of operator pencils, we can provide precise characterizations of F-optimal and minimal F-optimal POVMs.


\begin{proposition}\label{prop:FoptConPencilEig}
	Suppose $U$ is any unitary operator on $\caH$ that satisfies $\sqrt{\rho}\sqrt{\sigma}\lsp U = \sqrt{\sqrt{\rho}\lsp\sigma\sqrt{\rho}}$. Then $\scrE$ is F-optimal for distinguishing $\rho$ and $\sigma$ iff each   POVM element $E_m$ is supported in an eigenspace of $(\sqrt{\sigma},U\sqrt{\rho}\lsp)$ with a nonnegative eigenvalue (including $\infty$). If $\scrE$ is  F-optimal, then $\scrE$ is minimal iff no two POVM elements are supported in the same eigenspace of $(\sqrt{\sigma},U\sqrt{\rho}\lsp)$.
\end{proposition}

Suppose $\scrE$ is F-optimal for distinguishing $\rho$ and $\sigma$. Then, as a simple corollary of \pref{prop:FoptConPencilEig}, there exists a unique coarse graining of $\scrE$ up to relabeling that is F-optimal and minimal: this POVM is constructed by merging the POVM elements of $\scrE$ that are supported in each eigenspace, where merging is a special coarse graining that replaces two or more POVM elements with their sum.

\begin{proposition}\label{prop:FoptimalJM}
	Suppose $\scrE_1$ and $\scrE_2$ are two inequivalent minimal F-optimal POVMs for distinguishing $\rho$ and $\sigma$. Then $\scrE_1$ and $\scrE_2$ are not compatible. Given $0\leq p\leq 1$, let $\scrE(p)$ be the coarse graining of $p\scrE_1\sqcup (1-p)\scrE_2$ that is  F-optimal and minimal. Then $\scrE(p)$ for any two different values of $p$ are not equivalent and thus not compatible.
\end{proposition}
Notably, if the minimal F-optimal POVM for distinguishing $\rho$ and $\sigma$ is not unique, then there exist infinitely many (uncountable) minimal F-optimal POVMs, as illustrated in \fref{Fig:illustration}.


\emph{Insights from  geometric means}---\Psref{prop:FoptConParallel} and \ref{prop:FoptConPencilEig} are instructive for understanding (minimal) F-optimal POVMs, but eigenspaces of  an operator pencil are less tangible than the eigenspaces of a single operator. To  simplify these criteria  and to construct concrete minimal F-optimal POVMs, we need to introduce the concept of operator geometric means  (see \aref{App:GeomMeanProperty}), which are also helpful for studying quantum communication \cite{Muller2013GMEntropy,Matsumoto2015GMEntropy,Fang2021GMEntropy}. When $\rho$ is nonsingular, the \emph{geometric mean} of $\rho^{-1}$ and $\sigma$ is defined as follows \cite{Lawson2001GM,Bhatia2007GeoMean}:
\begin{equation}\label{eq:GeometricMean}
\caM\left(\rho^{-1},\sigma\right):=\sqrt{\rho^{-1}}\sqrt{\sqrt{\rho}\lsp\sigma\sqrt{\rho}}\sqrt{\rho^{-1}}.
\end{equation}
Here we do not assume that $\sigma$ is nonsingular, although this is a common assumption in some literature. If $\rho$ and $\sigma$ commute, then $\caM\left(\rho^{-1},\sigma\right)=\sqrt{\rho^{-1}}\sqrt{\sigma}$ may be regarded as the quotient of $\sqrt{\sigma}$ over $\sqrt{\rho}$, and the unitary operator $U$ in \pref{prop:FoptConPencilEig} can be chosen to be the identity. In this case, there is a simple connection between the eigenvalues (eigenspaces) of $\caM\left(\rho^{-1},\sigma\right)$ and the counterparts of $(\sqrt{\sigma},U\sqrt{\rho}\lsp)$, which intuitively explains why $\caM\left(\rho^{-1},\sigma\right)$ is relevant to studying F-optimal POVMs. 

\begin{figure}
	\centering
	\includegraphics[width=0.48\textwidth]{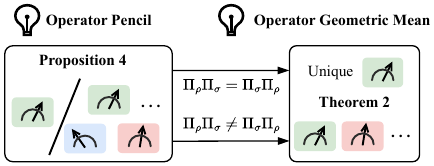}
	\caption{Weak-commutativity condition and dichotomy between a unique minimal F-optimal POVM and infinitely many inequivalent choices as represented by different colors.}
	\label{Fig:illustration}
\end{figure}

When $\rho$ is singular, the  inverse in \eref{eq:GeometricMean} can be replaced by the Moore-Penrose generalized inverse~\cite{Barata2012moore,Ben2002GInverse} of $\rho$, denoted by $\rho^+$ henceforth:
\begin{equation}\label{Eq:GeometricMeanGeneral}
\caM\left(\rho^+,\sigma\right):=\sqrt{\rho^{+}}\sqrt{\sqrt{\rho}\lsp\sigma\sqrt{\rho}}\sqrt{\rho^{+}}.
\end{equation}
Nevertheless, little is known about this generalization in the literature. Our technical contribution on this topic as presented in \aref{App:GeomMeanProperty} is crucial to understanding F-optimal POVMs. 
By definition we can deduce that $\caM\left(\rho^+,\sigma\right)$ is a positive operator with 
\begin{equation}
\supp\left(\caM\left(\rho^+,\sigma\right)\right)=\supp(\Pi_\rho\Pi_\sigma\Pi_\rho)\leq \supp(\rho),
\end{equation} 
where $\Pi_\rho$ and $\Pi_\sigma$ are the projectors onto $\supp(\rho)$ and $\supp(\sigma)$, respectively. 

Let $\Pi_{\rho,\sigma}$ be the projector onto $\supp(\Pi_\rho\Pi_\sigma\Pi_\rho)$. Define $\scrM(\rho,\sigma)$ as the PVM constructed from  $\bbone-\Pi_\rho$, $\Pi_\rho-\Pi_{\rho,\sigma}$, and the eigenprojectors of $\caM\left(\rho^+,\sigma\right)$ with positive eigenvalues; here the element $\bbone-\Pi_\rho$ ($\Pi_\rho-\Pi_{\rho,\sigma}$) will be omitted if its rank is 0, so that $\scrM(\rho,\sigma)$ is a simple PVM. Note that a POVM is a refinement of $\scrM(\rho,\sigma)$ iff each POVM element is supported in $\Null(\rho)$ or an eigenspace of $\caM(\rho^+,\sigma)$ within $\supp(\rho)$.
The geometric mean $\caM\left(\sigma^+,\rho\right)$ and the simple PVM $\scrM(\sigma,\rho)$ can be defined in a similar way. As we will see shortly, both $\scrM(\rho,\sigma)$ and $\scrM(\sigma,\rho)$ are minimal F-optimal POVMs for distinguishing $\rho$ and $\sigma$, which can also be deduced from  \rcite{Fuchs1995distinguishability}.
To clarify the significance of the geometric mean $\caM\left(\rho^+,\sigma\right)$ and the PVM $\scrM(\rho,\sigma)$, we need to introduce an auxiliary lemma. 

 
 \begin{lemma}\label{lem:Parallel}
 	Suppose $P$ is a  positive operator on $\caH$ that commutes with $\Pi_\rho$, and $U$ is any unitary operator that satisfies $\sqrt{\rho}\sqrt{\sigma}\lsp U = \sqrt{\sqrt{\rho}\lsp\sigma\sqrt{\rho}}$. Then  $P$ is supported in an eigenspace of $(\sqrt{\sigma},U\sqrt{\rho}\lsp)$
 iff $P$ is supported in $\Null(\rho)$ or an eigenspace of $\caM(\rho^+,\sigma)$ within $\supp(\rho)$. 
 \end{lemma}

Note that any coarse graining or refinement of $\scrM(\rho,\sigma)$ commutes with $\Pi_\rho$. In conjunction with \pref{prop:FoptConPencilEig} and \lref{lem:Parallel}, we can immediately derive the following theorem, which implies that $\scrM(\rho,\sigma)$ [$\scrM(\sigma,\rho)$] is the unique minimal F-optimal POVM up to relabeling that commutes with $\Pi_\rho$ ($\Pi_\sigma$).

\begin{theorem}\label{thm:FoptimalReg}
The PVM $\scrM(\rho,\sigma)$ is F-optimal and minimal  for distinguishing $\rho$ and $\sigma$.
Suppose  $\scrE$ commutes with $\Pi_\rho$; then $\scrE$ is F-optimal iff  it is a refinement of $\scrM(\rho,\sigma)$. The same conclusions hold if $\rho$ and $\sigma$ are exchanged.
\end{theorem}

Next, we present our main result on the connection between the uniqueness of the minimal F-optimal POVM and the weak commutativity condition. It turns out the same condition is also necessary and sufficient to guarantee that $\scrM(\rho,\sigma)$ and $\scrM(\sigma,\rho)$ are commuting, compatible, and equivalent.

\begin{theorem}\label{thm:FoptimalEqui}The five statements below are equivalent:
	\begin{enumerate}
		\item $\Pi_\rho$ and $\Pi_\sigma$ commute;
		
		\item $\scrM(\rho,\sigma)$ and $\scrM(\sigma,\rho)$ commute;

		\item $\scrM(\rho,\sigma)$ and  $\scrM(\sigma,\rho)$ are compatible;
		
		\item $\scrM(\rho,\sigma)$ and  $\scrM(\sigma,\rho)$ are equivalent;

        \item $\scrM(\rho,\sigma)$ is the unique minimal F-optimal POVM for distinguishing $\rho$ and $\sigma$    up to relabeling.
	\end{enumerate} 
\end{theorem}

\Thref{thm:FoptimalEqui} unifies five statements on commutativity, compatibility, equivalence, and uniqueness, which link algebraic and operational properties (see also \lref{lem:GIGMequi} in \aref{App:GeomMeanProperty}). When $\Pi_\rho$ and $\Pi_\sigma$ commute, which means each  principal angle between $\supp(\rho)$ and $\supp(\sigma)$ equals either 0 or $\pi/2$, all F-optimal POVMs for distinguishing $\rho$ and $\sigma$ are refinements of $\scrM(\rho,\sigma)$. Otherwise, by contrast, $\scrM(\rho,\sigma)$ and  $\scrM(\sigma,\rho)$ are two inequivalent minimal F-optimal POVMs, and we can construct infinitely many inequivalent alternative choices by \pref{prop:FoptimalJM}.

When $\rho$ is nonsingular, $\scrM(\rho,\sigma)$ is composed of the eigenprojectors of $\caM\left(\rho^{-1},\sigma\right)$. By virtue of \thref{thm:FoptimalReg} (and also \thref{thm:FoptimalEqui}) we can immediately deduce the following corollary, which is also implicit in \rcite{Fuchs1995distinguishability}. 
\begin{corollary}\label{cor:FoptimalInv}
	Suppose  $\rho$ is nonsingular. Then $\scrE$ is F-optimal for distinguishing $\rho$ and $\sigma$ iff each POVM element is supported in an eigenspace of $\caM(\rho^{-1},\sigma)$. It is F-optimal and minimal iff it is composed of the eigenprojectors of $\caM\left(\rho^{-1},\sigma\right)$.
\end{corollary}

So far, by virtue of operator pencils and geometric means, we have clarified the basic structure of minimal F-optimal POVMs, as illustrated in \fref{Fig:illustration}. In the rest of this paper, we offer additional insights on a special case in which all minimal F-optimal POVMs can be characterized completely.

\emph{F-optimal measurements for distinguishing a pure state and a mixed state}---Here we focus on the case in which one of the two states, say $\sigma$, is pure, assuming that $\rho+\sigma$ is nonsingular as before. Now, the minimal F-optimal PVM $\scrM(\sigma,\rho)=\{\bbone-\sigma,\sigma\}$ for distinguishing $\rho$ and $\sigma$ is  binary, while the minimal F-optimal PVM $\scrM(\rho,\sigma)$ is either binary or ternary. If  $\sigma$ commutes with $\Pi_\rho$, then $\scrM(\rho,\sigma)$ is equivalent to $\scrM(\sigma,\rho)$, and every F-optimal POVM is a refinement of $\scrM(\sigma,\rho)$ according to \thref{thm:FoptimalEqui}.

Next, we delve into the more interesting situation in which the pure state $\sigma$ does not commute with $\Pi_\rho$. Then $\rho$ is necessarily singular and $\dim(\Null(\rho))=1$. In addition, $\scrM(\rho,\sigma)$ and $\scrM(\sigma,\rho)$ are not equivalent by  \thref{thm:FoptimalEqui}, and there exist infinitely many inequivalent minimal F-optimal POVMs according to \pref{prop:FoptimalJM}. Is there a simple characterization of all such POVMs? 
To resolve this problem, here we introduce a simpler necessary and sufficient criterion on F-optimal POVMs (cf. \pref{prop:FoptConPencilEig}).

\begin{proposition}\label{pro:FoptimalPure1}
	Suppose  $\sigma$ is a pure state and $\rho\sigma\neq 0$.  Then  $\scrE$ is F-optimal for distinguishing $\rho$ and $\sigma$ iff each POVM element is supported in an eigenspace of the pencil $(\Pi_\rho\sigma, \Pi_\rho)$ with a nonnegative eigenvalue. 
\end{proposition}

Surprisingly, F-optimal POVMs for  $\rho$ and $\sigma$ are independent of $\rho$ once $\Pi_\rho$ and $\sigma$ are  fixed. Incidentally, the eigenspaces of $(\Pi_\rho\sigma, \Pi_\rho)$ with eigenvalues 0 and 1 coincide with $\Null(\Pi_\rho\sigma)=\Null(\sigma)$ and  $\supp(\sigma)$, respectively. 
The two eigenspaces are tied to the minimal F-optimal PVM $\scrM(\sigma,\rho)=\{\bbone-\sigma,\sigma\}$.

\begin{figure}
	\centering
	\includegraphics[width=0.4\textwidth]{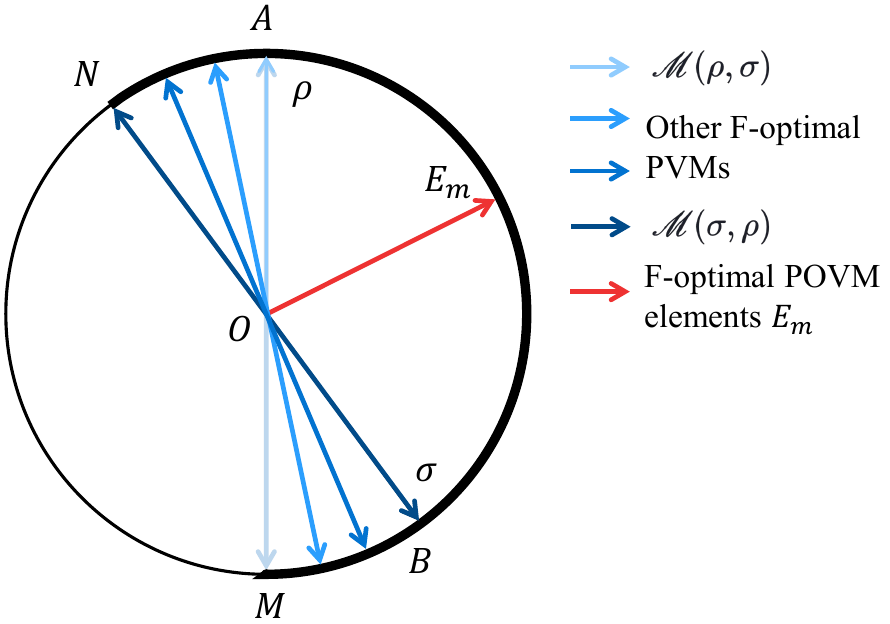}
	\caption{F-optimal POVMs for distinguishing two pure states $\rho$ and $\sigma$ illustrated  on the Bloch sphere.
Here  $A$ and $B$ denote the two pure states $\rho$ and $\sigma$,  while $M$ and $N$ denote their respective antipodal points. Each normalized POVM element of any F-optimal POVM is represented by a point on the major arc $\wideparen{NABM}$. }
	\label{Fig:purebloch}
\end{figure}

As an application of \pref{pro:FoptimalPure1}, first we analyze a simple but crucial case. Suppose $\rho$ is also a pure state, then
$\dim(\caH)=2$, and  both $\scrM(\rho,\sigma)$ and $\scrM(\sigma,\rho)$ are rank-1 PVMs.  In addition, each quantum state or normalized POVM element can be represented by a point on the Bloch sphere. Based on this picture, we can derive a more explicit criterion, which has a simple graphical representation on the Bloch sphere, as illustrated in \fref{Fig:purebloch}.

\begin{proposition}\label{pro:FoptimalPure2}
	Suppose $\dim \caH=2$ and $\rho,\sigma$ are two  pure states on $\caH$ with noncommuting density operators	
	represented by points $A$ and $B$ on the Bloch sphere;  let $M$ and $N$ be their respective antipodal points. Then $\scrE$ is F-optimal for distinguishing $\rho$ and $\sigma$ iff each POVM element is proportional to a pure state whose representative point lies on the major arc  $\wideparen{NABM}$. An F-optimal POVM is minimal iff it is simple. 
\end{proposition}

In the case of a qubit, if $\rho$ and $\sigma$  commute or if $\rho$ or $\sigma$ is nonsingular, then all simple F-optimal POVMs for $\rho$ and $\sigma$ are identical to $\scrM(\sigma,\rho)$ up to relabeling by \thref{thm:FoptimalEqui}. In conjunction with \pref{pro:FoptimalPure2} we can completely characterize F-optimal PVMs and POVMs for distinguishing two general single-qubit states.

Next, suppose $\rank(\rho)\geq 2$ and the pure state $\sigma$ does not commute with $\Pi_\rho$, which means $\dim(\caH)\geq 3$. As shown in \aref{app:FoptimalPureMix}, we can establish a one-to-one correspondence between F-optimal POVMs for (distinguishing) $\rho$ and $\sigma$ and F-optimal POVMs for $\varrho$ and $\sigma$, where $\varrho=\Pi_2\Pi_\rho\Pi_2$ is a pure state and $\Pi_2$ is the projector onto $\Null(\rho)+\supp(\sigma)$.  So this case  is no more difficult than the case of two pure states.


\emph{Summary}---We clarified the structure of optimal and minimal optimal quantum measurements for distinguishing two general quantum states with respect to the fidelity. When the projectors onto the two density operators commute, the minimal optimal measurement is uniquely determined by a geometric mean tied to the two density operators, and all optimal measurements are refinements of this special projective measurement. Otherwise, there exist infinitely many inequivalent minimal optimal measurements determined by the eigenspaces of a special operator pencil. 
Surprisingly, this dichotomy is tied to such a weak commutativity condition rather than the commutativity condition between the two density operators themselves as one might expect a priori. In addition, we showed that there exists  a unique  minimal optimal measurement that commutes with the projector onto the support of one of the two density operators. We also provided a complete characterization of all minimal optimal measurements in the case where one quantum state is pure. Our work sheds valuable insights on one of the most important similarity measures in quantum information science and a number of related topics, including quantum measurements and quantum incompatibility. Our results are amenable to experimental demonstration with current technologies in various platforms, including photonic, trapped ion, superconducting platforms; notably, it would be interesting to demonstrate inequivalent  minimal optimal measurements  in experiments. 
During our investigation, we derived a series of results on operator pencils and geometric means, which are of independent  interest. 



\let\oldaddcontentsline\addcontentsline
\renewcommand{\addcontentsline}[3]{}

\bibliography{reference}

\let\addcontentsline\oldaddcontentsline
\appendix

\tableofcontents

\bigskip

In this appendix, first we offer an illustrative example to support \thref{thm:FoptimalEqui} and provide more details on F-optimal measurements for distinguishing a pure state $\sigma$ and a mixed state $\rho$. Then we clarify F-optimal and minimal F-optimal measurements when $\rho+\sigma$ is singular. Next, we clarify minimal optimal measurements with respect to the trace distance for comparison. Subsequently, in preparation for proving our main results, we derive a number of auxiliary results on operator geometric means and operator pencils, which are of independent interest. Finally, we prove the results on F-optimal measurements as presented in the main text.

As in the main text,  $\caH$ denotes a $d$-dimensional Hilbert space, and $\bbone$ denotes the identity operator on $\caH$.  For any linear operator $A$ on $\caH$, $\Null(A)$ and $\supp(A)$ denote the null space and support of $A$, respectively. When $A$ is a positive operator, $\Pi_A$ denotes the projector onto the support of $A$. In addition, we denote by $\caP(\caH)$  the cone of positive (semidefinite) operators on~$\caH$.

\section{An example to support \thref{thm:FoptimalEqui}}\label{app:FoptExample}
When $\sigma$ commutes with $\Pi_\rho$, it is not difficult to find a unitary operator $U$ on $\caH$ that commutes with $\Pi_\rho$ and satisfies the condition $\sqrt{\rho}\sqrt{\sigma}\lsp U = \sqrt{\sqrt{\rho}\lsp\sigma\sqrt{\rho}}$ (see \lref{lem:AuxLemPolarUcom} in \aref{app:AuxLemmas}). Then it is straightforward to clarify the eigenvalues and eigenspaces of the pencil $(\sqrt{\sigma},U\sqrt{\rho})$, from which \thref{thm:FoptimalEqui} can be proved. When instead $\Pi_\sigma$ commutes with $\Pi_\rho$, but $\sigma$ does not commute with $\Pi_\rho$,
some equivalence relations  in \thref{thm:FoptimalEqui} are highly nontrivial. Here we offer a concrete example to help build intuition.

Consider the following two qutrit states:
\begin{equation}
	\begin{aligned}
		\rho&=\frac{1}{4}(2|0\>\<0|+|0\>\<1|+|1\>\<0|+2|1\>\<1|),\\\
		\sigma&=\frac{1}{4}(2|1\>\<1|+|1\>\<2|+|2\>\<1|+2|2\>\<2|),
	\end{aligned}
\end{equation}
with $\Pi_\rho=|0\>\<0|+|1\>\<1|$ and $\Pi_\sigma = |1\>\<1|+|2\>\<2|$. Note that $\Pi_\rho$ and $\Pi_\sigma$ commute with each other, and Statement 1 in  \thref{thm:FoptimalEqui} holds as desired. In addition, both $\caM(\rho^+,\sigma)$ and $\caM(\sigma^+,\rho)$ have an extremely simple form: $\caM(\rho^+,\sigma)=\caM(\sigma^+,\rho)=|1\>\<1|$, which means $\scrM(\rho,\sigma)=\scrM(\sigma,\rho)=\{|0\>\<0|,|1\>\<1|,|2\>\<2|\}$ up to relabeling. So the first four statements in \thref{thm:FoptimalEqui} hold.

Next, we  choose a specific unitary operator $U$  that satisfies  $\sqrt{\rho}\sqrt{\sigma}\lsp U = \sqrt{\sqrt{\rho}\lsp\sigma\sqrt{\rho}}$. Let $a=\sqrt{\sqrt{3}+2}$ and $b=3\sqrt{3}+5$, then a desired unitary can be expressed with respect to the computational basis as follows:
\begin{equation}
	\renewcommand{\arraystretch}{1.1}
	U = \frac{1}{4a^2}\begin{pmatrix}
		-2a^3 & 2a & 0 \\
		a^2 & a^4  & -2a \\
		1 & a^2 & 2a^3 \\
	\end{pmatrix}.
\end{equation}
In addition, we have
\begin{equation}
	\renewcommand{\arraystretch}{1.1}
	U\sqrt{\rho} = \frac{1}{8a^2}\begin{pmatrix}
		-2\sqrt{6}\lsp a^2 & 0 & 0 \\
		b &  2b  & 0 \\
		\sqrt{2}\lsp a  &  2\sqrt{2}\lsp a & 0 \\
	\end{pmatrix}.
\end{equation}
Now, it is straightforward to determine all eigenvalues  of $(\sqrt{\sigma},U\sqrt{\rho})$ and the corresponding normalized eigenvectors (up to overall phase factors), with the result
\begin{equation}
	\begin{aligned}
		\lambda_0 &= 0,&\quad &|\lambda_0\> = |0\>,\\
		\lambda_1 &= 1,&\quad& |\lambda_1\> = |1\>,\\
		\lambda_2 &= \infty,&\quad& |\lambda_2\>=|2\>.
	\end{aligned}
\end{equation}
According to \pref{prop:FoptConPencilEig}, 
every F-optimal POVM for distinguishing $\rho$ and $\sigma$ is equivalent to and also a refinement of $\scrM(\rho,\sigma)$. So $\scrM(\rho,\sigma)$ is the unique minimal F-optimal POVM for distinguishing $\rho$ and $\sigma$ up to relabeling, which confirms \thref{thm:FoptimalEqui} for this specific case.

\section{\label{app:FoptimalPureMix}F-optimal measurements for distinguishing a pure state and a mixed state}
Here we provide more details on F-optimal measurements for distinguishing a pure state $\sigma$ and a mixed state $\rho$ with $\rank(\rho)\geq 2$, where $\sigma$ does not commute with $\Pi_\rho$, which means $\dim(\caH)\geq 3$. We will reduce this case  to the case of two pure states clarified in the main text.

Let $\caV=\Null(\rho)+\supp(\sigma)$ and $\Pi_2$ be the projector onto $\caV$; then $\dim(\caV)=2$ and $\Pi_2$ commutes with $\Pi_\rho$. 
Let $\varrho=\Pi_2\Pi_\rho\Pi_2$; then $\varrho$ is a pure state supported in $\caV$
and satisfies $\varrho\sigma=\Pi_\rho\sigma$. Based on this observation, we can construct  F-optimal POVMs for distinguishing $\rho$ and $\sigma$ from the counterpart for $\varrho$ and $\sigma$, and vice versa.

Specifically, if $\scrE$ is an F-optimal POVM for  $\rho$ and $\sigma$, then $\Pi_2 \scrE \Pi_2$ is  an F-optimal POVM for $\varrho$ and $\sigma$. If in addition $\scrE$ is minimal, then  $\Pi_2 \scrE \Pi_2$ is minimal after deleting a zero POVM element if necessary. Conversely, if $\scrE'$ is a minimal F-optimal POVM for  $\varrho$ and $\sigma$, then $\scrE'\cup \{\bbone-\Pi_2\}$ is a minimal F-optimal POVM for distinguishing $\rho$ and $\sigma$ after merging POVM elements supported in $\Null(\sigma)$ if necessary. In this way, we establish a one-to-one correspondence between minimal F-optimal POVMs for $\rho$ and $\sigma$ and the counterpart for $\varrho$ and $\sigma$ (and similarly for PVMs). Since F-optimal POVMs for distinguishing $\varrho$ and $\sigma$ are completely characterized by \pref{pro:FoptimalPure2}, the structure of F-optimal POVMs for $\rho$ and $\sigma$ is also clear. In contrast with the qubit case, now $\scrM(\sigma,\rho)$ is the unique binary F-optimal POVM up to relabeling; any other F-optimal POVM, including $\scrM(\rho,\sigma)$, contains at least three POVM elements.

Next, we clarify the reason underlying the above one-to-one correspondence. If we regard $(\varrho\sigma, \varrho)$ as a pencil on $\caV$, then $(\varrho\sigma, \varrho)$ and $(\Pi_\rho\sigma, \Pi_\rho)$ have the same eigenvalues. Given a complex number $\lambda$ (including $\infty$), let $\Pi_\lambda$ be the eigenprojector of $(\Pi_\rho\sigma, \Pi_\rho)$ with eigenvalue $\lambda$ and let $\Pi_\lambda'$ be the counterpart of $(\varrho\sigma, \varrho)$. Then we have the following simple correspondence:
\begin{align}\label{eq:EigMap}
	\Pi_\lambda=\Pi_\lambda'+(\bbone-\Pi_2)\delta_{\lambda,0}=\begin{cases}
		\bbone-\sigma &\lambda = 0,\\
		\Pi'_\lambda&\lambda\neq0.
	\end{cases}
\end{align}
Notably, if $\lambda\neq 0$, then the eigenspace of $(\Pi_\rho\sigma, \Pi_\rho)$ with eigenvalue $\lambda$ is a one-dimensional subspace of $\caV$ and coincides with the eigenspace of  $(\varrho\sigma, \varrho)$ with eigenvalue $\lambda$.

\section{F-optimal measurements when $\rho+\sigma$ is singular}

Now,  $\rho$ and $\sigma$ can also be regarded as density operators on $\supp(\rho+\sigma)$. Not surprisingly, a POVM on $\caH$ is F-optimal for distinguishing $\rho$ and $\sigma$ iff its restriction on $\supp(\rho+\sigma)$ is F-optimal, as shown in the following proposition. 

\begin{proposition}\label{Prop:Projection}
	Suppose $\rho, \sigma$ are two quantum states on $\caH$ and $\scrE=\{E_m\}_m$ is a POVM on $\caH$.  Let $\Pi$ be the projector onto $\supp(\rho+\sigma)$ and 	 
	\begin{equation}\label{Eq:Projection}
		\scrE'=\Pi\scrE\Pi:=\{\Pi E_m \Pi\}_m, 
	\end{equation}
	regarded as a POVM on the support of $\rho+\sigma$.
	Then $\scrE$ is F-optimal for distinguishing $\rho$ and $\sigma$ iff $\scrE'$ is F-optimal for distinguishing $\rho$ and $\sigma$. In addition, if $\scrE$ is F-optimal and minimal, then $\scrE'$ is F-optimal and minimal.
\end{proposition}

 When $\rho+\sigma$ is singular, \pref{prop:FoptConParallel} in the main text still holds. In addition, the two PVMs  $\scrM(\rho,\sigma)$ and $\scrM(\sigma,\rho)$  are still F-optimal for distinguishing $\rho$ and $\sigma$, but they cannot be equivalent. If $\supp(\sigma)\leq \supp(\rho)$, then  $\scrM(\rho,\sigma)$ is not minimal; if $\supp(\rho)\leq \supp(\sigma)$, then $\scrM(\sigma,\rho)$ is not minimal. If neither $\supp(\sigma)\leq \supp(\rho)$ nor $\supp(\rho)\leq \supp(\sigma)$, then both  $\scrM(\rho,\sigma)$ and $\scrM(\sigma,\rho)$ are  F-optimal and minimal, but they are not equivalent.

 Whenever $\rho+\sigma$ is singular, there exist infinitely many inequivalent minimal F-optimal POVMs for distinguishing $\rho$ and $\sigma$. To be concrete, let  $\scrC=\{C_j\}_j$ be a minimal F-optimal POVM on $\supp(\rho+\sigma)$ and let $\Gamma_0$ be the projector onto $\Null(\rho+\sigma)$. Then, for any probability distribution $\{a_j\}_j$ on $\scrC$, the POVM  $\{C_j+a_j\Gamma_0\}_j$ on $\caH$ is F-optimal and minimal;  moreover, any two distinct POVMs in this family are inequivalent. On the other hand, all these POVMs lead to the same classical probability distribution when performed on $\rho$ or $\sigma$, so they are essentially the same for distinguishing $\rho$ and~$\sigma$.

In the rest of this appendix, we prove \pref{Prop:Projection}.
\begin{proof}[Proof of \pref{Prop:Projection}]
	By assumption, it is straightforward to verify that $F_{\scrE'}(\rho,\sigma)
	=F_\scrE(\rho,\sigma)$. Therefore, $\scrE$ is F-optimal for distinguishing $\rho$ and $\sigma$ iff $\scrE'$ is F-optimal for distinguishing $\rho$ and $\sigma$.

	Next, suppose $\scrE=\{E_m\}_m$ is F-optimal and minimal and is thus automatically simple. Then $\scrE'=\{E'_m\}_m$ with $E_m'=\Pi E_m\Pi$ is F-optimal as shown above. In addition, $\scrE'$ is simple. To prove this point, suppose, by way of contradiction, that $\scrE'$ is not simple. Then there exist two POVM elements in $\scrE'$, say $E'_1$ and $E'_2$, that are linearly dependent. Let $\scrC$ be the POVM constructed from $\scrE$ by replacing the two POVM elements $E_1$ and $E_2$ with their sum $E_1+E_2$. Then $\scrC$ is a nontrivial coarse graining of $\scrE$ that is also F-optimal, which contradicts the assumption that $\scrE$ is F-optimal and minimal. This contradiction shows that $\scrE'$ is simple.

	Now, let $\scrB'=\{B'_j\}_j$
	be a coarse graining of $\scrE'$ that is F-optimal for distinguishing $\rho$ and $\sigma$. Then $B'_j$ can be expressed as follows:
	\begin{align}
		\!\!	B_j' = \sum_k S_{jk}E'_k =\sum_k S_{jk}\Pi E_k\Pi =\Pi \sum_k S_{jk} E_k\Pi,
	\end{align}
	where $S$ is a stochastic matrix. Construct the POVM  $\scrB = \{B_j\}_j$ with  $B_j = \sum_k S_{jk} E_k$; then  $\scrB$ is  a coarse graining of $\scrE$.
	By definition, we have
	\begin{align}
		F_{\scrB}(\rho,\sigma)=F_{\scrB'}(\rho,\sigma)=F(\rho,\sigma),
	\end{align}
	so the POVM $\scrB$ is also F-optimal for distinguishing $\rho$ and $\sigma$. Since $\scrE$ is F-optimal and  minimal by assumption, $\scrB$ is necessarily equivalent to $\scrE$, which means $\scrB'$ is equivalent to $\scrE'$.  Therefore, any nontrivial coarse graining of $\scrE'$ is not  F-optimal. It follows that $\scrE'$ is F-optimal and minimal given that $\scrE'$ is  F-optimal and simple as shown above. This observation completes the proof of \pref{Prop:Projection}.
\end{proof}

\section{Optimal measurements with respect to the trace distance}\label{app:Toptimal}

The trace distance is a popular measure for quantifying the distinction between two quantum states \cite{Fuchs1996Measures,Fuchs1999Measures,Nielsen2010Quantum}. To complement the results presented in the main text, here we briefly discuss optimal and minimal optimal measurements with respect to the trace distance.

\subsection{Basic properties}

Let $\bm{p}\equiv\{p_m\}_m$ and $\bm{q}\equiv\{q_m\}_m$ be two probability distributions indexed by a same set.  
The (classical) trace distance  between $\bm{p}$ and $\bm{q}$ reads 
 \begin{equation}\label{Eq:ClassicalTraceDistance}
 D(\bm{p},\bm{q})= \frac{1}{2}\sum_m|p_m-q_m|
=\max_{\bbS}\sum_{m \in \bbS} (p_m-q_m), 
\end{equation}
where the maximization is taken over all subsets of the index set. Suppose the two probability distributions are chosen with probability $1/2$ each, then the maximum probability of successfully distinguishing between the two distributions is given by
\begin{equation}
P_\rms=\frac{1+D(\bm{p},\bm{q})}{2}.
\end{equation}
A larger trace distance corresponds to a higher success probability, highlighting the operational significance of the trace distance.

Next, we turn to the quantum analog. 	Suppose $\rho, \sigma$ are two quantum states on $\caH$ and $\scrA, \scrB, \scrE=\{E_m\}_m$ are three POVMs on $\caH$. The trace distance between $\rho$ and $\sigma$  reads~\cite{Nielsen2010Quantum}
\begin{equation}
D(\rho,\sigma) = \frac{1}{2}\tr\left|\rho-\sigma\right|.
\end{equation}
The trace distance and fidelity are closely tied to each other and satisfy the well-known Fuchs-van de Graaf inequalities~\cite{Fuchs1999Measures,Nielsen2010Quantum}:
\begin{equation}
    1-\sqrt{F(\rho,\sigma)}\le D(\rho,\sigma)\le\sqrt{1-F(\rho,\sigma)}.
\end{equation}
To understand the connection with the classical counterpart, denote by $D_\scrE(\rho,\sigma)$ the classical trace distance between the probability distributions $\Lambda_\scrE(\rho)$ and $\Lambda_\scrE(\sigma)$ induced by the POVM $\scrE$. Then  $D_\scrE\left(\rho,\sigma\right)\leq D(\rho,\sigma)$ according to Helstrom-Holevo theorem, and the inequality is tight \cite{Helstrom1969Estimation,Nielsen2010Quantum, Fuchs1999Measures}. So the trace distance $D(\rho,\sigma)$ determines how well the two states $\rho$ and $\sigma$ can be distinguished. The POVM $\scrE$ is \emph{T-optimal} for distinguishing $\rho$ and $\sigma$  if the inequality is saturated, that is, $D_\scrE\left(\rho,\sigma\right)=D(\rho,\sigma)$.  
Similar to the fidelity, the induced classical trace distance is monotonic under data processing.
\begin{proposition}\label{prop:DCoarse}
Suppose $\scrA$ is a coarse graining of $\scrB$. Then
\begin{equation}\label{eq:DCoarse}
D_\scrA(\rho,\sigma)\le D_\scrB(\rho,\sigma).
\end{equation}
\end{proposition}
A T-optimal POVM  is \emph{minimal} if it is simple and any nontrivial coarse graining is not T-optimal.

To characterize T-optimal POVMs,  note that $\rho-\sigma$ can be expressed as
 \begin{align}
 \rho-\sigma=Q_+-Q_-,
\end{align}
where $Q_+$ and $Q_-$ are two positive  operators that are mutually orthogonal, that is, $Q_+Q_-=0$, and are called the positive part and negative part of $\rho-\sigma$, respectively. Let $\Pi_+$ and $\Pi_-$
be the projectors onto the supports of $Q_+$ and $Q_-$, respectively, and let $\Pi_0$ be the projector onto the null space of $\rho-\sigma$. 
The following proposition clarifies  T-optimal POVMs for distinguishing $\rho$ and $\sigma$ \cite{Nielsen2010Quantum}. 
\begin{proposition}\label{prop:Toptimal} 
The POVM $\scrE$ is T-optimal for distinguishing $\rho$ and $\sigma$ iff each POVM element $E_m$ is supported in $\Null(Q_+)$ or $\Null(Q_-)$, that is,
 \begin{equation}\label{eq:Toptimal}
 Q_+ E_m = 0\quad \mathrm{or}\quad Q_- E_m =0 \quad\forall m,
\end{equation}
where $Q_+$ and $Q_-$ denote the positive part and negative part of $\rho-\sigma$, respectively.
\end{proposition}

When $\rho$ and $\sigma$ are distinct single-qubit states, both $\Pi_+$ and $\Pi_-$ are rank-1 eigenprojectors of $\rho-\sigma$, and $Q_+$ and $Q_-$ are proportional to  $\Pi_+$ and $\Pi_-$, respectively. Therefore, all simple T-optimal POVMs for distinguishing $\rho$ and $\sigma$ are equivalent to the PVM $\{\Pi_+,\Pi_-\}$.

\begin{proposition}\label{pro:ToptimalMinimal}
The POVM $\scrE$ is T-optimal and minimal for distinguishing $\rho$ and $\sigma$ iff it has two POVM elements 
and has the form $\{\Pi_++Q_0,\bbone-\Pi_+-Q_0\}$ up to relabeling, where $\Pi_+$ and $\Pi_-$
are the projectors onto the supports of the positive part and negative part of $\rho-\sigma$, respectively, and $0\leq Q_0\leq \bbone$ is supported in $\Null(\rho-\sigma)$. 
\end{proposition}

According to \pref{pro:ToptimalMinimal}, if $\rho-\sigma$ is nonsingular, then  $\{\Pi_+,\Pi_-\}$ is the unique minimal T-optimal POVM up to relabeling for distinguishing $\rho$ and $\sigma$; otherwise,  there exist infinitely many inequivalent minimal T-optimal POVMs. For instance, the POVM 
\begin{equation}
\scrA=\{\Pi_++\lambda\Pi_{0},\Pi_-+(1-\lambda)\Pi_{0}\}
\end{equation}
is T-optimal and minimal for each $\lambda$ in the interval $[0,1]$. Moreover, any two distinct POVMs in this family are inequivalent. However, these POVMs commute and are compatible, in sharp contrast with the results on minimal F-optimal POVMs  (cf. \pref{prop:FoptimalJM} and  \thref{thm:FoptimalEqui}).

\subsection{Proofs of \psref{prop:DCoarse}-\ref{pro:ToptimalMinimal}}
\begin{proof}[Proof of \pref{prop:DCoarse}]
Suppose $\scrA=\{A_j\}_j$ and $\scrB=\{B_k\}_k$. 
Denote the probability vectors $\Lambda_\scrA(\rho)$, $\Lambda_\scrA(\sigma)$, $\Lambda_\scrB(\rho)$, $\Lambda_\scrB(\sigma)$ by $\bmp, \bmq, \bmp', \bmq'$, respectively.
By assumption, $\scrA$ is a coarse graining of $\scrB$, so each POVM element $A_j$ of $\scrA$ can be expressed as  $A_j = \sum_k S_{jk} B_k$,
where $S$ is a stochastic matrix. Accordingly, we have
\begin{equation}
	p_j = \sum_k S_{jk}p_k',\quad
	q_j = \sum_k S_{jk} q_k'.
\end{equation}
Therefore,
\begin{align}
	D_\scrA(\rho,\sigma)&=\frac{1}{2}\sum_{j}|p_j-q_j|= \frac{1}{2}\sum_{j}\left|\sum_{k} S_{jk} (p'_k-q'_k)\right|\nonumber\\
	&\le \frac{1}{2}\sum_{j,k} S_{jk}\left|p'_k-q'_k\right|=D_\scrB(\rho,\sigma),
\end{align}
where the inequality follows from the triangle inequality and the last equality holds because $S$ is a stochastic matrix, which means $\sum_{j}S_{jk}=1$ for all $k$. The above equation confirms \eref{eq:DCoarse} and completes the proof of \pref{prop:DCoarse}.
\end{proof}

\begin{proof}[Proof of \pref{prop:Toptimal}]
By definition ,we can deduce that
\begin{align}
&D_\scrE\left(\rho,\sigma\right)=\frac{1}{2}\sum_m\left|\tr\left[E_m(\rho-\sigma)\right]\right|\nonumber\\
&=\frac{1}{2}\sum_m\left|\tr\left[E_m(Q_+-Q_-)\right]\right|\le\frac{1}{2}\sum_m \tr\left[E_m(Q_++Q_-)\right]\nonumber\\
&=\frac{1}{2}\sum_m\tr(E_m|\rho-\sigma|)=\frac{1}{2}\tr|\rho-\sigma|=D(\rho,\sigma),
\end{align}
where the first equality holds because   $\rho-\sigma=Q_+-Q_-$. The inequality follows from the triangle inequality and is saturated iff $\tr(Q_+ E_m)=0$ or $\tr(Q_+ E_m)=0$ for each $m$, which is the case iff \eref{eq:Toptimal} holds, that is, each  POVM element $E_m$ is supported in $\Null(Q_+)$ or $\Null(Q_-)$. This observation completes the proof of \pref{prop:Toptimal}.
\end{proof}


\begin{proof}[Proof of \pref{pro:ToptimalMinimal}]
First, by way of contradiction, we will prove that the number of  POVM elements in any minimal T-optimal POVM cannot surpass 2. Suppose $\scrE$ is a minimal T-optimal POVM for distinguishing $\rho$ and $\sigma$ and has at least three POVM elements. According to  \pref{prop:Toptimal}, each POVM element in $\scrE$ is supported in either $\Null(Q_+)$ or $\Null(Q_-)$, and there exist at least two elements that are supported in  $\Null(Q_+)$ simultaneously or in $\Null(Q_-)$ simultaneously. 
Therefore, we can construct a nontrivial coarse graining of $\scrE$  by merging these POVM elements. The resulting POVM is still  T-optimal, which contradicts the assumption that $\scrE$ is T-optimal and minimal. This contradiction shows that any minimal T-optimal POVM has only two POVM elements.

Next, suppose $\scrE$ is a binary POVM that is T-optimal and minimal. Then, according to \pref{prop:Toptimal},  one POVM  element is supported  $\Null(Q_-)$, and the other is supported in $\Null(Q_+)$. So $\scrE$  necessarily has the form  $\scrE=\{\Pi_++Q_0,\bbone-\Pi_+-Q_0\}$ up to relabeling, where $Q_+$ and $Q_-$ denote the positive part and negative part of $\rho-\sigma$, respectively, and $Q_0$ is a positive operator supported in $\Null(\rho-\sigma)$  and satisfies $0\leq Q_0\leq \bbone$. This observation completes the proof of \pref{pro:ToptimalMinimal}.
\end{proof}

\section{\label{app:AuxLemmas}Auxiliary lemmas}

\begin{lemma}\label{lem:AuxLemPQ}
Suppose $P$ and $Q$ are two projectors on $\caH$, and $\Pi_{PQP}$ and $\Pi_{QPQ}$ are projectors onto $PQP$ and $QPQ$, respectively. Then the following four statements are equivalent.
	\begin{enumerate}
	\item $P$ and $Q$ commute;
	
	\item $\supp(PQ P)=\supp(P)\cap\supp(Q)$;
	
	\item $\Null(PQ P)=\Null(P)+\Null(Q)$;
	
	\item $PQ P$ and $QP Q$ have the same support. 
	
	\item 	$PQ P$ and $QP Q$ commute;
	
	\item 	$\Pi_{PQP}$ and $\Pi_{QPQ}$ commute.
\end{enumerate}
\end{lemma}

\begin{lemma}\label{lem:AuxLemABsupp}
	Suppose $A,B\in \caP(\caH)$. Then $\supp(AB)\leq \supp(A)$ iff $B$ commutes with $\Pi_A$. In addition, $AB$ and $BA$  have the same support iff  $A$ commutes with $\Pi_B$ and $B$ commutes with $\Pi_A$. 
\end{lemma}

\begin{lemma}\label{lem:AuxLemPolarUcom}
    Suppose $A,B\in\caP(\caH)$. Then there exists a unitary operator  $U$ on $\caH$ that commutes with $\Pi_A$ and satisfies  $\sqrt{A}\sqrt{B}\lsp U=\sqrt{\sqrt{A}\lsp B\sqrt{A}}$  iff $B$ commutes with~$\Pi_A$.
\end{lemma}

\begin{proof}[Proof of \lref{lem:AuxLemPQ}]
For any projectors $P$ and $Q$, there exists two orthogonal basis $\{|\psi_i\>\}_i$ and $\{|\varphi_i\>\}_i$ such that
\begin{equation}
\begin{gathered}
P=\sum_{i=1}^{\rank P} |\psi_i\>\<\psi_i|,\quad Q=\sum_{j=1}^{\rank Q}|\varphi_j\>\<\varphi_j|,\\
|\<\psi_i|\varphi_j\>|=q_i \delta_{i,j}\quad\forall i\le\rank P, j\le\rank Q.
\end{gathered}
\end{equation}
Therefore, it suffices to prove the case when $\rank P=\rank Q=1$. In this case, it is direct to show that the six statements in \lref{lem:AuxLemPQ} are equivalent.
\end{proof}

\begin{proof}[Proof of \lref{lem:AuxLemABsupp}]
If $\supp(AB)\le\supp(A)$, then we have $AB\Pi_A=AB$. After multiplying both sides of the equation by $A^+$ from the left side, we can deduce that
\begin{align}
\Pi_AB\Pi_A=\Pi_AB=B\Pi_A,
\end{align}
so $B$ commutes with $\Pi_A$. Here the last equality holds because $\Pi_AB\Pi_A$ is Hermitian. Conversely, if $B$ commutes with $\Pi_A$, then  $AB\Pi_A=AB$, which implies that  $\supp(AB)\le \supp(A)$. Therefore,  $\supp(AB)\leq \supp(A)$ iff $B$ commutes with $\Pi_A$. By symmetry we can deduce that   $\supp(BA)\leq \supp(B)$ iff $A$ commutes with $\Pi_B$

Next, we turn to the second statement in \lref{lem:AuxLemABsupp}. 
If $AB$ and $BA$  have the same support, then $\supp(AB)\leq \supp(A)$ and $\supp(BA)\leq \supp(B)$. So $A$ commutes with $\Pi_B$ and $B$ commutes with $\Pi_A$ according to the above discussion. Conversely, if  $A$ commutes with $\Pi_B$ and $B$ commutes with $\Pi_A$, then
\begin{align}
\supp(AB)=\supp(BA)=\supp(A)\cap\supp(B),
\end{align}
which completes the proof of \lref{lem:AuxLemABsupp}.
\end{proof}

\begin{proof}[Proof of \lref{lem:AuxLemPolarUcom}]
By assumption we can deduce that $\supp\left(\sqrt{A}\lsp B\sqrt{A}\lsp\right)\le\supp(\sqrt{A}\lsp)=\supp(A)$. 	
If  $B$ commutes with $\Pi_A$, then $\sqrt{B}$ also commutes with $\Pi_A$, and $\supp\left(\sqrt{A}\sqrt{B}\lsp\right)\le \supp(A)$ by \lref{lem:AuxLemABsupp}. So  there exists a unitary operator  $U$ on $\caH$ that commutes with $\Pi_A$ and satisfies  $\sqrt{A}\sqrt{B}\lsp U=\sqrt{\sqrt{A}\lsp B\sqrt{A}}$.

Conversely, if there exists a unitary operator  $U$ that commutes with $\Pi_A$ and satisfies  $\sqrt{A}\sqrt{B}\lsp U=\sqrt{\sqrt{A}\lsp B\sqrt{A}}$, then $U^\dag$ also  commutes with $\Pi_A$ and we can deduce that  $\supp\left(\sqrt{A}\sqrt{B}\lsp\right)\leq \supp(A)$. Therefore, $\sqrt{B}$ commutes with $\Pi_{\sqrt{A}}$, which means $B$ commutes with $\Pi_A$ and completes the proof of \lref{lem:AuxLemPolarUcom}.
\end{proof}

\section{Properties of  operator geometric means}\label{App:GeomMeanProperty}

Before proving the results on the optimal measurements with respect to the fidelity as presented in the main text, here we summarize the key properties of the geometric mean of two positive operators $A$ and $B$ on $\caH$. When $A, B$ are nonsingular (or positive-definite), their geometric mean, denoted by $\caM(A, B)$, is defined as
\begin{equation}\label{eq:DefGeoMean}
	\caM(A, B):= \sqrt{A} \sqrt{\sqrt{A^{-1}}\lsp B\sqrt{A^{-1}}} \sqrt{A}.
\end{equation}
This geometric mean has the following properties~\cite{Lawson2001GM,Bhatia2007GeoMean}:
\begin{enumerate}
	\item $\caM(A,B)$ is positive definite.
	\item $\caM(A,B)^{-1} = \caM\left(A^{-1}, B^{-1}\right)$.
	\item $\caM(A,B)=\caM(B,A)$.
	\item $\caM(A,B)$ is the unique positive solution of the matrix equation $XA^{-1}X=B$. Notably, $\caM(A,B)=\bbone$ iff $A=B^{-1}$.
\end{enumerate}
Note that \eref{eq:DefGeoMean} is still well-defined when $B$ is singular as long as $A$ is nonsingular. When  $A$ is singular, we can extend the above definition by virtue of the Moore-Penrose generalized inverse (MPGI)~\cite{Barata2012moore,Ben2002GInverse}.

Generally, for any linear operator $O$ on $\caH$, the Moore-Penrose generalized inverse of $O$, denoted by $O^+$, is uniquely determined by the following four conditions:
\begin{enumerate}
	\item $OO^+O=O$;
	\item $O^+OO^+ = O^+$;
	\item $\left(OO^+\right)^\dagger = OO^+$;
	\item $\left(O^+O\right)^\dagger = O^+O$.
\end{enumerate}
By definition we have $(O^+)^+=O$; in addition, 
$OO^+$ and $O^+O$ are projectors onto $\range(O)$ and $\range(O^+)$, respectively. If $O$ is invertible, then the MPGI of $O$ coincides with its inverse $O^{-1}$.  If $O$ is normal, that is, $O^\dag O=O O^\dag$, then $O^+$ has  the same support as $O$ and is the inverse of $O$ in $\supp(O)$, that is, $OO^+=O^+O=\Pi_O$, where $\Pi_O$ is the projector onto the support of $O$. If in addition $O$ is positive semidefinite, then $(O^x)^+=(O^+)^x$ and $(O^+)^x O^x =O^x (O^+)^x=\Pi_O$
for $x>0$.

Now, the (generalized) geometric mean of two general positive operators $A, B$ on $\caH$ can be defined as follows: 
\begin{equation}
	\caM(A,B):= \sqrt{A}\sqrt{\sqrt{A^+}\lsp B\sqrt{A^+}}\sqrt{A},
\end{equation} 
which is a positive operator by definition. 
Next, we summarize the basic properties of the geometric mean of two positive operators.  

\subsection{Basic properties}

\begin{lemma}\label{lem:GMsuppNull}
	Suppose $A,B\in\caP(\caH)$,  $\Pi_A$ and $\Pi_B$ are the projectors onto the supports of $A$ and $B$, respectively, 
	and $\caM = \caM(A,B)$. Then 
	\begin{gather}
		\Null(\caM)=\Null(\Pi_A\Pi_B \Pi_A), \label{eq:NullM}\\ \supp(\caM)=\supp(\Pi_A\Pi_B \Pi_A), \label{eq:suppM}\\
		\Null(A)\leq \Null(\caM) \leq  \Null(A)+\Null(B), \label{eq:NullMrho+}\\
		\supp(A)\cap \supp(B)\leq \supp(\caM)\leq \supp(A). \label{eq:suppMrho+}
	\end{gather}
\end{lemma}

\begin{lemma}\label{lem:GMeig}
	Suppose $A, B, P\in\caP(\caH)$ and $P$ is supported in $\supp(A)\cap\supp(B)$ and an eigenspace of $\caM(A, B)$  with eigenvalue $\lambda$. Then $\lambda>0$ and $P$ is supported in the eigenspace of $\caM\left(B^+, A^+\right)$  with eigenvalue $1/\lambda$. The same result holds if  $\caM(A, B)$ and $\caM(B^+, A^+)$ are exchanged.
\end{lemma}

\begin{lemma}\label{lem:GIGMequi}
	Suppose $A,B\in\caP(\caH)$, and $\Pi_A$ and $\Pi_B$ are the projectors onto the supports of $A$ and $B$, respectively. Then the following nine statements are equivalent:
	\begin{enumerate}
		\item  $\Pi_A$ and $\Pi_B$ commute;
		
		\item $\supp(\caM(A,B))=\supp(A)\cap\supp(B)$;
		
		\item  $\supp(\caM(A,B))\leq \supp(B)$;
		
		\item $\Null(\caM(A,B))=\Null(A)+\Null(B)$;
		
		\item $\Null(B)\leq \Null(\caM(A,B))$;
				
		\item $\caM(B,A)$ and  $\caM(A,B)$ have the same support;
		
		\item  $\caM(B^+,A^+)$ and  $\caM(A,B)$ have the same support;
		
		\item $\caM(B^+,A^+)$ is the MPGI of  $\caM(A,B)$. 

        \item $\caM(B^+,A^+)$ and $\caM(A,B)$ commute.

	\end{enumerate}
\end{lemma}

\begin{lemma}\label{lem:GMcommu}
	Suppose $A,B\in\caP(\caH)$. Then $\caM(A,B)=\caM(B,A)$ iff $\caM(B^+,A^+)=\caM(A^+,B^+)$. 
\end{lemma}

\begin{lemma}\label{lem:GMcommu2}
	Suppose $A,B\in\caP(\caH)$,  $\Pi_A$ and $\Pi_B$ are the projectors onto the supports of $A$ and $B$, respectively, $A$ commutes with $\Pi_B$, and $B$ commutes with $\Pi_A$. Then $\caM(A,B)=\caM(B,A)$ and $\caM(B^+,A^+)=\caM(A^+,B^+)$. 
\end{lemma}

Next, we turn to the special case in which one operator, say $A$, is nonsingular.
The following two lemmas are simple corollaries of \lsref{lem:GMsuppNull} and \ref{lem:GIGMequi}, respectively. 
\begin{lemma}\label{lem:GeoMeanNullSpace}
	Suppose $A,B\in\caP(\caH)$ and $A$ is nonsingular. Then the null spaces (supports) of $\caM(A,B)$, $\caM(B, A)$, $\caM(A,B^+)$, and $\caM(B^+,A)$ all coincide with the null space (support) of $B$,
	and the same results still hold if $A$ is replaced by $A^{-1}$.       
\end{lemma}

\begin{lemma}\label{lem:GMEquiv}
	Suppose $A, B\in\caP(\caH)$ and $A$ is nonsingular. Then $\caM\left(B^+, A^{-1}\right)$ and $\caM\left(A^{-1},B^+\right)$ are the  MPGIs of $\caM(A, B)$ and $\caM(B, A)$, respectively.	

\end{lemma}

\subsection{Proofs of \lsref{lem:GMsuppNull}-\ref{lem:GMcommu2}}

\begin{proof}[Proof of \lref{lem:GMsuppNull}]Let $R=\sqrt{A^+}\lsp B\sqrt{A^+}$, then 
	\begin{align}
		\caM=\sqrt{A}\sqrt{R}\sqrt{A}=\left(R^{1/4}\sqrt{A}\lsp \right)^\dag R^{1/4}\sqrt{A}. 
	\end{align}	
	Therefore,
	\begin{align}
		&\Null(\caM)=\Null\left(R^{1/4}\sqrt{A}\lsp\right)=\Null\left(R\sqrt{A}\lsp\right)\nonumber\\
		&=\Null\bigl(\sqrt{A^+}\lsp B\Pi_A \bigr)=\Null(\Pi_A B\Pi_A)=\Null\bigl(\sqrt{B}\lsp\Pi_A\bigr)\nonumber\\
		&=\Null(\Pi_B \Pi_A)=\Null(\Pi_A\Pi_B \Pi_A),
	\end{align}	
	which confirms \eref{eq:NullM} and implies \eref{eq:suppM}. \Eqsref{eq:NullMrho+}{eq:suppMrho+} are simple corollaries of 	\eqsref{eq:NullM}{eq:suppM}.
\end{proof}

\begin{proof}[Proof of \lref{lem:GMeig}]
	Note that $\caM(A, B)$ is a positive operator and $\supp(A)\cap\supp(B)\leq \supp(\caM(A, B))$ according to \lref{lem:GMsuppNull}.  In addition, $P$ is supported in $\supp(A)\cap\supp(B)$ by assumption and is thus supported in  $\supp(\caM(A, B))$. Now, suppose $P$ is supported in an eigenspace of $\caM(A, B)$  with eigenvalue $\lambda$, then we have $\lambda>0$ because this eigenspace is necessarily contained in $\supp(\caM(A, B))$.

	Let $U$ be any unitary on $\caH$ that satisfies the condition 
	\begin{align}\label{eq:GMeigProof}
		\sqrt{A^+}\sqrt{B}\lsp U = U^\dag \sqrt{B}\sqrt{A^+}= \sqrt{\sqrt{A^+}B\sqrt{A^+}}.
	\end{align}
	Then
	\begin{align}
		 U^\dag \sqrt{B}\lsp  P &=  \sqrt{\sqrt{A^+}\lsp B\sqrt{A^+}}\sqrt{A}\lsp P
		\nonumber\\
		&=\sqrt{A^+}\sqrt{A}\sqrt{\sqrt{A^+}\lsp B\sqrt{A^+}}\sqrt{A}\lsp P
		\nonumber\\
		&= \sqrt{A^+}\lsp\caM(A,B) P=
		\lambda\sqrt{A^+}\lsp P, \label{eq:EigEqProof1}
	\end{align}
	given that $P$ is supported in $\supp(A)\cap\supp(B)$ and an eigenspace of $\caM(A, B)$  with eigenvalue $\lambda>0$. In addition, $\sqrt{A}\sqrt{A^+}$ and $\sqrt{A^+}\sqrt{A}$ coincide with the projector onto  $\supp(A)$; similarly, $\sqrt{B}\sqrt{B^+}$ and $\sqrt{B^+}\sqrt{B}$ coincide with the projector onto $\supp(B)$.

	Next, multiplying \eref{eq:EigEqProof1} with $  \sqrt{B^+}\lsp U$ (on the left) yields
	\begin{align}
		&P = 
		\lambda  \sqrt{B^+}\lsp U \sqrt{A^+}\lsp P=\lambda  \sqrt{B^+}\lsp U \sqrt{A^+}\sqrt{B}\sqrt{B^+}\lsp P
		\nonumber\\
		&= \lambda \caM\left(B^+, A^+\right) P,
	\end{align}
	where the last equality follows from the equation below,
	\begin{align}
		\sqrt{B} \sqrt{A^+}\lsp U^\dagger  =U\sqrt{A^+}\sqrt{B}=		
		\sqrt{\sqrt{B}A^+\sqrt{B}},
	\end{align}
	which, in turn, follows from \eref{eq:GMeigProof}. So $P $ is supported in the eigenspace of $\caM(B^+,A^+)$ with eigenvalue $1/\lambda $. 
	
	Finally, suppose instead $P $ is supported in the eigenspace of $\caM(B^+,A^+)$ with eigenvalue $\lambda $.  Note that $A^+$ and $B^+$ have the sample supports as $A$ and $B$, respectively; in addition, $(A^+)^+=A$ and   $(B^+)^+=B$. Therefore, $\lambda>0$ and $P $ is supported in the eigenspace of $\caM(A,B)$ with eigenvalue $1/\lambda $ according to a similar reasoning employed above, which completes the proof of  \lref{lem:GMeig}. 
\end{proof}

\begin{proof}[Proof of \lref{lem:GIGMequi}]By assumption $A$ and $B$ are positive operators, so  $A^+$ and $B^+$ have the same supports as $A$ and $B$, respectively. In conjunction with \lref{lem:GMsuppNull} we can deduce that
	\begin{equation}\label{eq:GIGMequiProof}
		\begin{aligned}
			\supp(\caM(A,B))&=\supp(\Pi_A\Pi_B \Pi_A),\\ \supp\left(\caM\left(B^+,A^+\right)\right)&=\supp(\caM(B,A))\\
			&=\supp(\Pi_B\Pi_A \Pi_B). 
		\end{aligned}	
	\end{equation}
	Therefore, the equivalence of Statements 1-7 in \lref{lem:GIGMequi} follows from \lref{lem:AuxLemPQ}.

	Next, suppose the first seven statements in \lref {lem:GIGMequi} hold. Then the supports of $\caM(B^+,A^+)$ and $\caM(A,B)$ coincide with $\supp(A)\cap\supp(B)$. 
	Moreover, according to \lref{lem:GMeig}, $\caM(B^+,A^+)$ is the inverse of $\caM(A,B)$ when restricted to $\supp(A)\cap\supp(B)$. So  $\caM(A^+,B^+)$ is the MPGI of  $\caM(A,B)$, which confirms Statement 8.  Conversely, if Statement 8 holds, then $\caM(B^+,A^+)$ and  $\caM(A,B)$ necessarily have the same support,  which confirms Statement 7,  given that $\caM(B^+,A^+)$ and  $\caM(A,B)$
	are positive operators. So the first eight statements in \lref{lem:GIGMequi} are equivalent.
    
    Finally, suppose the first eight statements in \lref{lem:GIGMequi} hold. Then $\caM(B^+,A^+)$ is the MPGI of $\caM(A,B)$ and thus commutes with $\caM(A,B)$, which confirms Statement~9. Conversely, if $\caM(B^+,A^+)$ and $\caM(A,B)$ commute, then the projectors onto their supports necessarily commute. In conjunction with \eref{eq:GIGMequiProof} and  \lref{lem:AuxLemPQ}  we can deduce that $\Pi_A$ and $\Pi_B$ commute, which confirms Statement 1. 
    In summary, all nine statements in \lref{lem:GIGMequi} are equivalent, which completes the proof.
\end{proof}

\begin{proof}[Proof of \lref{lem:GMcommu}]
	Suppose $\caM(A,B)=\caM(B,A)$; then $\caM(A,B)$ and $\caM(B,A)$ have the same supports. According to \lref{lem:GIGMequi}, $\caM(B^+,A^+)$ and $\caM(A^+,B^+)$ are the MPGIs of $\caM(A,B)$ and $\caM(B,A)$, respectively, which implies that $\caM(B^+,A^+)=\caM(A^+,B^+)$. Conversely, if $\caM(B^+,A^+)=\caM(A^+,B^+)$, then $\caM(A,B)=\caM(B,A)$ by a similar reasoning given that $(A^+)^+=A$ and $(B^+)^+=B$. 
\end{proof}

\begin{proof}[Proof of \lref{lem:GMcommu2}]
	By assumption, $B$ can be expressed in the form $B=B_0+B_1$, where $B_0$ and $B_1$ are positive operators supported in $\Null(A)$ and $\supp(A)$, respectively. In addition, $\Pi_B=\Pi_{B_0}+\Pi_{B_1}$ and $A$ commutes with both $\Pi_{B_0}$ and $\Pi_{B_1}$. It follows that  $A$ can be expressed as $A=A_0+A_1$, where $A_0$ and $A_1$ are positive operators supported in $\supp(A)\cap\Null(B_1)$ and $\supp(B_1)=\supp(A)\cap\supp(B)$, respectively; meanwhile, $A_1$ has full rank within  $\supp(B_1)$. Therefore, 
	\begin{align}
		\caM(A,B)=\caM(A_1,B_1)=\caM(B_1,A_1)=\caM(B,A),
	\end{align}
	where the second equality holds because the MPGI is symmetric in the two operators when both operators have the same support and are nonsingular within the support.

	In addition, $A^+$ and $B^+$ have the same supports and the same sets of eigenspaces as $A$ and $B$, respectively. Furthermore,  $\Pi_B$ commutes with $A$ iff it commutes with $A^+$; similarly, $\Pi_A$ commutes with $B$ iff it commutes with $B^+$. Therefore,  $\caM(B^+,A^+)=\caM(A^+,B^+)$ by a similar reasoning employed above. Alternatively, this conclusion follows from the equality $\caM(A,B)=\caM(B,A)$ proved above thanks to \lref{lem:GMcommu}.
\end{proof}

\section{\label{sec:Pencil}Properties of operator pencils}

In this section we clarify the basic properties of linear operator pencils that are helpful for understanding F-optimal measurements. Here we assume that $A$ and $B$ are nonzero positive operators in $\caP(\caH)$ and $A+B$ is nonsingular. In \lsref{lem:PencilEigSpace}-\ref{lem:PencilEigWC} and \coref{cor:PencilEigANS} below, we further assume that  $U$ is a unitary operator that satisfies the condition  $\sqrt{A}\sqrt{B}\lsp U = \sqrt{\sqrt{A}\lsp B\sqrt{A}}$.

\subsection{Basic properties}

Although eigenvectors of the pencil $(\sqrt{B},U\sqrt{A}\lsp)$ associated with distinct eigenvalues are not necessarily linearly independent, any two such eigenvectors are linearly independent according to the following lemma. 

\begin{lemma}\label{lem:PencilEigSpace}
	Any two distinct eigenspaces of the pencil $(\sqrt{B},U\sqrt{A}\lsp)$ have a trivial intersection. If ${P,Q\in\caP(\caH)}$ are nonzero positive operators  supported in distinct eigenspaces of $(\sqrt{B},U\sqrt{A}\lsp)$. Then $P+Q$ is not supported any eigenspace.
\end{lemma}

\begin{lemma}\label{lem:PencilEigSpacePcom}
	Suppose $P\in \caP(\caH)$ is a  positive operator that commutes with $\Pi_A$  and is supported in an eigenspace of $(\sqrt{B},U\sqrt{A}\lsp)$. Then $P$ is supported in either $\supp(A)$ or $\Null(A)$.
\end{lemma}

Next, we clarify the relations between the eigenspaces of the pencil $(\sqrt{B},U\sqrt{A}\lsp)$ and the eigenspaces of the geometric mean $\caM=\caM(A^+,B)$. 
Recall that  the eigenspaces of $(\sqrt{B},U\sqrt{A}\lsp)$ with eigenvalues $\infty$ and 0 coincide with  $\Null(A)$ and  $\Null(B)$, respectively,  and $\Null(\caM)=\Null(\Pi_A\Pi_B \Pi_A)$ by \lref{lem:GMsuppNull}. If $\Pi_A$ and $\Pi_B$ commute, then 
$\Null(A)\bot \Null(B)$ and $\Null(\caM)=\Null(A)\oplus \Null(B)$ by \lref{lem:GIGMequi}, so $\Null(\caM)$ is the direct sum of the two eigenspaces of $(\sqrt{B},U\sqrt{A}\lsp)$ with eigenvalues $\infty$ and 0.

\begin{lemma}\label{lem:PencilEigGMsuppA}
Any vector in   $\supp(A)$  is an  eigenstate of $(\sqrt{B},U\sqrt{A}\lsp)$  
	iff it is an  eigenstate of $\caM(A^+,B)$ with the same eigenvalue. 
\end{lemma}

If $A$ is nonsingular, then  $\supp(A)=\caH$  and all eigenvalues of $(\sqrt{B},U\sqrt{A}\lsp)$ are finite. In this case, \lref{lem:PencilEigGMsuppA} implies the following result as a simple corollary.
\begin{corollary}\label{cor:PencilEigANS}
	Suppose  $A$ is nonsingular. Then any vector in   $\caH$  is an  eigenstate of $(\sqrt{B},U\sqrt{A}\lsp)$  
	iff it is an  eigenstate of $\caM(A^{-1},B)$ with the same eigenvalue. 
\end{corollary}

\begin{lemma}\label{lem:PencilEigWC}
	Suppose  $B$ and $U$ commute with $\Pi_A$. Then every eigenstate of  $(\sqrt{B},U\sqrt{A}\lsp)$ belongs to either  $\Null(A)$ or $\supp(A)$.
\end{lemma}

When $B$ commutes with~$\Pi_A$, thanks to \lref{lem:AuxLemPolarUcom},  there exists a unitary operator  $U$ on $\caH$ that commutes with $\Pi_A$ and satisfies  $\sqrt{A}\sqrt{B}\lsp U=\sqrt{\sqrt{A}\lsp B\sqrt{A}}$.
Furthermore, the commutativity condition between $B$ and  $\Pi_A$ in \lref{lem:PencilEigWC} can be  relaxed as shown in the following lemma.
\begin{lemma}\label{lem:PencilEigPiC}
	Suppose $\Pi_A$ and $\Pi_B$ commute. Then there exists a unitary operator $U$ on $\caH$ that satisfies  $\sqrt{A}\sqrt{B}\lsp U = \sqrt{\sqrt{A}\lsp B\sqrt{A}}$, and every eigenvector of $(\sqrt{B}, U\sqrt{A}\lsp )$ belongs to either $\Null(A)$ or $\supp(A)$.
\end{lemma}

The following lemma is a simple corollary of \lsref{lem:PencilEigGMsuppA} and \ref{lem:PencilEigPiC}.
\begin{lemma}\label{lem:PencilEigGM}
	Suppose $\Pi_A$ and $\Pi_B$ commute. Then there exists a unitary operator $U$ on $\caH$ that satisfies  $\sqrt{A}\sqrt{B}\lsp U = \sqrt{\sqrt{A}\lsp B\sqrt{A}}$, and any nonzero vector in $\caH$ is an eigenstate of $(\sqrt{B},U\sqrt{A}\lsp)$ iff it belongs to $\Null(A)$ or an eigenspace of $\caM(A^+,B)$ within $\supp(A)$. 
\end{lemma}
If in addition  $B$ commutes with $\Pi_A$, then there exists  a unitary operator $U$ that commutes with $\Pi_A$ and  satisfies the condition $\sqrt{A}\sqrt{B}\lsp U = \sqrt{\sqrt{A}\lsp B\sqrt{A}}$. Moreover,  such a unitary operator automatically satisfies the requirement in \lref{lem:PencilEigGM} 
 according to \lsref{lem:PencilEigGMsuppA} and \ref{lem:PencilEigWC}.

\subsection{Proofs of \lsref{lem:PencilEigSpace}-\ref{lem:PencilEigPiC}}

\begin{proof}[Proof of \lref{lem:PencilEigSpace}]
	Suppose, by way of contradiction, that $\caV_1$ and $\caV_2$ are two eigenspaces of $(\sqrt{B},U\sqrt{A}\lsp)$ with different eigenvalues $\lambda_1$ and $\lambda_2$ that have a nontrivial intersection. Let $|\psi\>$ be a normalized vector in  $\caV_1\cap\caV_2$. If $\lambda_1=\infty$, then $\lambda_2\neq \infty$ and  
	\begin{equation}
		\begin{aligned}
			|\psi\>&\in \Null\bigl(U\sqrt{A}\lsp\bigr)=\Null\bigl(\sqrt{A}\lsp\bigr)=\Null(A),\\ |\psi\>&\in \Null\bigl(\sqrt{B}-\lambda_2 U\sqrt{A}\lsp\bigr),
		\end{aligned}
	\end{equation}
	which means $|\psi\>\in \Null(A)\cap\Null(B)=\Null(A+B)$, in contradiction with the assumption that $A+B$ is nonsingular. This contradiction shows that $\lambda_1 \neq \infty$ and $\lambda_2\neq\infty$ by symmetry. Consequently,
	\begin{equation}
		|\psi\>\in \Null\bigl(\sqrt{B}-\lambda_1 U\sqrt{A}\lsp\bigr)\cap \Null\bigl(\sqrt{B}-\lambda_2 U\sqrt{A}\lsp\bigr),
	\end{equation}
	which means $|\psi\>\in \Null(\sqrt{A}\lsp)\cap\Null(\sqrt{B}\lsp)=\Null(A+B)$. This contradiction shows that any two distinct eigenspaces of $(\sqrt{B},U\sqrt{A}\lsp)$ have a trivial intersection.

	Next, suppose $P,Q\in\caP(\caH)$ are nonzero and supported in distinct eigenspaces of $(\sqrt{B},U\sqrt{A}\lsp)$. Suppose, by way of contradiction, that $P+Q$ is supported in an eigenspace of $(\sqrt{B},U\sqrt{A}\lsp)$. Then $P$ and $Q$ are supported in the same eigenspace, which contradicts the conclusion proved above. This contradiction completes the proof of \lref{lem:PencilEigSpace}.
\end{proof}

\begin{proof}[Proof of \lref{lem:PencilEigSpacePcom}] If $P$ is supported in $\Null(A)$, then there is nothing to prove, so we assume that $P$  is not supported in $\Null(A)$ in this proof. By assumption   $P$ commutes with $\Pi_A$ and thus can be expressed as $P=P_0+P_1$, where $P_0$ and $P_1\neq 0$ are supported in $\Null(A)$ and $\supp(A)$, respectively, so that $\sqrt{A}\lsp P=\sqrt{A}\lsp P_1\neq 0$. In addition, by assumption $\sqrt{B}\lsp P$ and $U\sqrt{A}\lsp P$ are parallel, which means
	\begin{align}\label{eq:ParallelConProof}
		\sqrt{B}\lsp P=\sqrt{B}\lsp P_0+\sqrt{B}\lsp P_1=\lambda U\sqrt{A}\lsp P =\lambda U\sqrt{A}\lsp P_1,
	\end{align}
	where $\lambda$ is some nonnegative constant. This equation can hold only if $\sqrt{B}\lsp P_0=0$, so $P_0$ is necessarily supported in $\Null(A)\cap\Null(B)=\Null(A+B)$, which means $P_0=0$ given that  $A+B$ is nonsingular by assumption. Therefore,  $P=P_1$ is supported in $\supp(A)$, which completes the proof of \lref{lem:PencilEigSpacePcom}.  
\end{proof}

\begin{proof}[Proof of \lref{lem:PencilEigGMsuppA}]
Without loss of generality, let  $|\psi\>$ be a normalized vector in $\supp(A)$.  Then  $\Pi_A|\psi\>=|\psi\>$, where $\Pi_A$  is the projector onto $\supp(A)$ and can be expressed as $\Pi_A=
	\sqrt{A}\sqrt{A^+}=\sqrt{A^+}\sqrt{A}$. In conjunction with the assumption $\sqrt{A}\sqrt{B}\lsp U = \sqrt{\sqrt{A}\lsp B\sqrt{A}}$ we can deduce that
	\begin{equation}\label{eq:PencilEigSuppProof1}
		\begin{aligned}
			\sqrt{\sqrt{A}\lsp B\sqrt{A}}&=U^\dag\sqrt{B}\sqrt{A},\\
			\sqrt{\sqrt{A}\lsp B\sqrt{A}}\sqrt{A^+}\lsp |\psi\>
			&=U^\dag\sqrt{B}\lsp \Pi_A |\psi\>=U^\dag\sqrt{B}\lsp |\psi\>,\\
			\caM\left(A^+,B\right)|\psi\>&=\sqrt{A^+}\lsp U^\dag\sqrt{B}\lsp |\psi\>, 
		\end{aligned}	
	\end{equation}

	Suppose $|\psi\>$  is an  eigenstate of $(\sqrt{B},U\sqrt{A}\lsp)$ with eigenvalue $\lambda$. Then
	$\lambda\neq \infty$ given that $|\psi\>\in \supp(A)$ by assumption, and we have 
	\begin{equation}\label{eq:ParallelProof1}
		U^\dag\sqrt{B}\lsp |\psi\>=\lambda \sqrt{A}\lsp  |\psi\>.
	\end{equation}
	The above equations together imply that
	\begin{align}
		\caM\left(A^+,B\right)|\psi\>=\lambda \sqrt{A^+}\sqrt{A}\lsp |\psi\>=\lambda |\psi\>.
	\end{align}
	So 	$|\psi\>$  is an  eigenstate of $\caM(A^+,B)$ with eigenvalue $\lambda$.

	Conversely, suppose  $|\psi\>$  is an  eigenstate of $\caM(A^+,B)$ with eigenvalue $\lambda$, then  $0\leq \lambda<\infty$, given that $\caM(A^+,B)$ is a positive operator. In conjunction with \eref{eq:PencilEigSuppProof1} we can deduce that
	\begin{align}
		U^\dag\sqrt{B}\lsp |\psi\>&= \sqrt{A}\sqrt{A^+} \sqrt{\sqrt{A}\lsp B\sqrt{A}}\sqrt{A^+}\lsp |\psi\>\nonumber\\
		&= \sqrt{A}\lsp\caM\left(A^+,B\right) |\psi\>=\lambda  \sqrt{A}\lsp |\psi\>.
	\end{align}
	So  $|\psi\>$  is an  eigenstate of $(\sqrt{B},U\sqrt{A}\lsp)$ with eigenvalue $\lambda$, which  completes the proof of \lref{lem:PencilEigGMsuppA}. 
\end{proof}

\begin{proof}[Proof of \lref{lem:PencilEigWC}]
	Let $\bPi_A=\bbone-\Pi_A$ be the projector onto $\Null(A)$.	By assumption $B$ and $U$ commute with $\Pi_A$ and  $\bPi_A$, so  $\sqrt{B}$ also  commutes with $\Pi_A$ and  $\bPi_A$.
	
	Suppose $|\psi\>$ is a normalized eigenstate of $(\sqrt{B},U\sqrt{A}\lsp)$ with eigenvalue $\lambda$, and let $|\psi_0\>=\bPi_A|\psi\>$, which is not necessarily normalized. If $|\psi\>\in \Null(A)$, then there is nothing to prove, so we assume that $|\psi\>\notin\Null(A)$, which means $\lambda\neq \infty$ and 
	\begin{align}
		\sqrt{B}\lsp |\psi\>=\lambda U\sqrt{A}\lsp |\psi\>.
	\end{align}
	Multiplying both sides of the above equation by the projector $\bPi_A$  yields $\sqrt{B}\lsp |\psi_0\>=0$, which implies that $|\psi_0\>\in \Null(A)\cap\Null(B)=\Null(A+B)$. 
	Since  $A+B$ is nonsingular by assumption, we can conclude that  $|\psi_0\>=0$ and $|\psi\>\in \supp(A)$, which completes the proof of \lref{lem:PencilEigWC}.
\end{proof}

\begin{proof}[Proof of \lref{lem:PencilEigPiC}]
By assumption $A$ and $B$ are positive operators on $\caH$, $A+B$ is nonsingular, and  $\Pi_A$ commutes with $\Pi_B$. If $A$ is nonsingular, then the conclusion in \lref{lem:PencilEigPiC} is obvious. So we assume that $A$ is nonsingular in the following proof. Then $\caH$ can be expressed as follows:
\begin{align}\label{eq:HDecom}
	\caH &=\Null(A)\oplus\supp(A)=\Null(B)\oplus\supp(B)\nonumber\\
	&= \Null(A)\oplus\Null(B)\oplus[\supp(A)\cap\supp(B)],
\end{align}
where  $\Null(A)$, $\Null(B)$, and $\supp(A)\cap\supp(B)$ are mutually orthogonal, and $\Null(B)$ can be omitted if $B$ is nonsingular.  Let $\bPi_A=\bbone-\Pi_A$ be the projector onto $\Null(A)$,  $r=\tr(\Pi_A\Pi_B)=\dim(\supp(A)\cap\supp(B))$ and $\caM=\caM(A^+,B)$.
By virtue of \lref{lem:GIGMequi} we can deduce that  $\caM^+=\caM(B^+,A)$ and
\begin{align}\label{eq:rankComSupp}
	\rank\left(\sqrt{A}\sqrt{B}\lsp\right)&=\rank\left(\caM\right)=\rank(\Pi_A\Pi_B)=r.
\end{align}

Next, we turn to the singular-value decomposition of $\sqrt{A}\sqrt{B}$ in order to find a desired unitary operator $U$ featured in the polar decomposition. In view of \eref{eq:rankComSupp},  $\sqrt{A}\sqrt{B}$ has a singular-value decomposition of the form
\begin{equation}\label{eq:sssrSVD}
\sqrt{A}\sqrt{B} = \sum_{i=1}^{r} s_i|\alpha_i\>\<\beta_i|,
\end{equation}
where $\{s_i\}_{i=1}^r$ is the set of  nonzero singular values, $\{|\alpha_i\>\}_{i=1}^r$ is a set of orthonormal vectors in $\supp(A)$,
and $\{|\beta_i\>\}_{i=1}^r$ is a set of orthonormal vectors in $\supp(B)$. Thanks to \eref{eq:HDecom}, we can extend $\{|\alpha_i\>\}_{i=1}^r$ and $\{|\beta_i\>\}_{i=1}^r$ into orthonormal bases $\{|\alpha_i\>\}_{i=1}^d$ and $\{|\beta_i\>\}_{i=1}^d$ on $\caH$ with the following properties:  
\begin{equation}\label{eq:alphabetaBasis}
\begin{aligned}
	|\alpha_i\>\in\Null(A),\;\; |\beta_i\>\in\supp(B)\;  &\mbox{  if  }\; r<i\le r_B,\\
	|\alpha_i\>\in\supp(A),\;\; |\beta_i\>\in\Null(B) \; &\mbox{  if  } \; r_B< i\le d,
\end{aligned}
\end{equation}
where $r_B=\rank(B)$ and the second case can occur only if $B$ is singular. By construction we can deduce the following results:
\begin{equation}\label{eq:spanssf}
\begin{gathered}
	\supp(B) =\mathrm{span}\left\{|\beta_i\>\right\}_{i=1}^{r_B}= \mathrm{span}\left\{\sqrt{B^+}\lsp|\beta_i\>\right\}_{i=1}^{r_B},\\
	\supp\left(\sqrt{\sqrt{B}A\sqrt{B}}\lsp\right)=\supp\left(\sqrt{B}A\sqrt{B}\lsp \right) =\mathrm{span}\left\{|\beta_i\>\right\}_{i=1}^{r},\\
	\supp(\Pi_A\Pi_B)=\supp(\caM^+)=\mathrm{span}\left\{\sqrt{B^+}\lsp|\beta_i\>\right\}_{i=1}^r,\\
	\Null(A)=\mathrm{span}\left\{\sqrt{B^+}\lsp|\beta_i\>\right\}_{i=r+1}^{r_B}.
\end{gathered}
\end{equation}

Now, a desired unitary operator $U$ that satisfies the condition $\sqrt{A}\sqrt{B}\lsp U = \sqrt{\sqrt{A}\lsp B\sqrt{A}}$ can be constructed as follows:
\begin{equation}
	U=\sum_{i=1}^d|\beta_i\>\<\alpha_i|.
\end{equation}
 Suppose $|\psi\>$ is an eigenvector of $(\sqrt{B},U\sqrt{A}\lsp)$ with eigenvalue $\lambda$. If $\lambda=\infty$, then $|\psi\>\in\Null(A)$ by definition. Otherwise, $|\psi\>\notin\Null(A)$ and  $\left(\sqrt{B}-\lambda U\sqrt{A}\lsp \right)|\psi\> =0$. Consequently,  
\begin{align}
0& = \left(U^\dag\sqrt{B}-\lambda\sqrt{A}\lsp \right)|\psi\>\nonumber\\
&=\left(U^\dag\sqrt{B}\sqrt{A}\sqrt{A^+}-\lambda\sqrt{A}\lsp \right)|\psi\> + U^\dag\sqrt{B}\lsp\bPi_A|\psi\>\nonumber\\
&=\sqrt{A}\lsp \left(\caM-\lambda \bbone\right)|\psi\>+U^\dag\sqrt{B}\lsp\bPi_A|\psi\>,  \label{eq:PencilEigPiCproof}
\end{align}
where the second equality holds because $\bbone=\Pi_A+\bPi_A$ and $\Pi_A=\sqrt{A}\sqrt{A^+}$, and the third equality holds because 
\begin{align}
	&U^\dag\sqrt{B}\sqrt{A}\sqrt{A^+}=\sqrt{\sqrt{A}B\sqrt{A}}\sqrt{A^+}\nonumber\\
	&=\sqrt{A}\sqrt{A^+}\sqrt{\sqrt{B}A\sqrt{B}}\sqrt{A^+}=\sqrt{A}\lsp\caM. 
\end{align}

According to \eref{eq:spanssf}, $\bPi_A|\psi\>$ in \eref{eq:PencilEigPiCproof} can be expressed as follows:
\begin{align}
\bPi_A|\psi\>= \sum_{i=r+1}^{r_B} a_i\sqrt{B^+}\lsp|\beta_i\>,
\end{align}
where $a_i$ are complex coefficients. In conjunction with \eref{eq:alphabetaBasis} we can deduce  that
\begin{align}
U^\dag\sqrt{B}\lsp\bPi_A|\psi\>&=U^\dag\sum_{i=r+1}^{r_B}a_i|\beta_i\>=\sum_{i=r+1}^{r_B}a_i|\alpha_i\>\in\Null(A),
\end{align}
where the first equality holds because $\sqrt{B}\sqrt{B^+}=\Pi_B$ and $|\beta_i\>\in \supp(B)$ for all $i=1,2,\ldots, r_B$. Meanwhile, $\sqrt{A}\lsp \left(\caM-\lambda \bbone\right)|\psi\>\in \supp(A)$. So \eref{eq:PencilEigPiCproof} implies that $U^\dag\sqrt{B}\lsp\bPi_A|\psi\>$, which means $a_i=0$ for $r+1\le i\le r_B$ and $\bPi_A|\psi\>=0$. Therefore, $|\psi\>\in\supp(A)$, which  completes the proof of \lref{lem:PencilEigPiC}.
\end{proof}

\section{\label{app:FoptimalProofs}Proofs of general results on F-optimal measurements}

\subsection{Proof of \pref{pro:CoarseIneq}}
\begin{proof}
Suppose $\scrA=\{A_j\}_j$ and $\scrB=\{B_k\}_k$. 
Denote the probability vectors $\Lambda_\scrA(\rho)$, $\Lambda_\scrA(\sigma)$, $\Lambda_\scrB(\rho)$, $\Lambda_\scrB(\sigma)$ by $\bmp, \bmq, \bmp', \bmq'$, respectively.
By assumption, $\scrA$ is a coarse graining of $\scrB$. So we have
\begin{equation}
	p_j = \sum_k S_{jk}p_k',\quad
	q_j = \sum_k S_{jk} q_k',
\end{equation}
where $S$ is a stochastic matrix. Therefore,
\begin{align}
	&\sqrt{F_\scrA(\rho,\sigma)}=\sum_{j}\sqrt{p_j q_j}= \sum_{j}\sqrt{\sum_kS_{jk}p'_k\sum_k S_{jk}q'_k}\nonumber\\
	&\ge \sum_{j}\sum_k S_{jk}\sqrt{p_k'q_k'}= \sum_k \sqrt{p_k'q_k'}=\sqrt{F_\scrB(\rho,\sigma)},
\end{align}
which confirms \eref{eq:CoarseIneq}. Here the inequality follows from the Cauchy-Schwarz inequality and the third equality holds because $S$ is a stochastic matrix, which means $\sum_{j}S_{jk}=1$ for all $k$. 
\end{proof}

\subsection{Proof of \pref{prop:FoptConParallel}}
\begin{proof}[Proof of \pref{prop:FoptConParallel}]
	By definition, we have
	\begin{align}
		&\sqrt{F_\scrE(\rho,\sigma)}=\sum_m\sqrt{\tr(E_m\rho)\tr(E_m\sigma)}\nonumber\\
		&\ge \sum_m \Re\tr\left[\bigl(U\sqrt{\rho}\sqrt{E_m}\lsp\bigr)\bigl(\sqrt{\sigma}\sqrt{E_m}\lsp\bigr)^\dagger\right]\nonumber\\
		&=\sum_m \Re \tr\left(\sqrt{\rho}\lsp E_m\sqrt{\sigma}\lsp U\right)=\tr\left(\sqrt{\rho}\sqrt{\sigma}\lsp U\right)\nonumber\\
		&=\left\|\sqrt{\rho}\sqrt{\sigma}\lsp\right\|_1=\sqrt{F(\rho,\sigma)}\lsp.
 	\end{align}
	Here the  inequality follows from the Cauchy-Schwarz inequality, and it is saturated iff each POVM element $E_m$ satisfies the following condition:
	\begin{equation}
\sqrt{\rho}\sqrt{E_m} = 0\quad\text{or}\quad  \sqrt{\sigma}\sqrt{E_m}=\kappa_mU\sqrt{\rho}\sqrt{E_m},
	\end{equation}
	where $\kappa_m$ is a nonnegative
	constant, that is, $\sqrt{\sigma}\sqrt{E_m}$ and $U\sqrt{\rho}\sqrt{E_m}$ are parallel. This observation completes the proof of \pref{prop:FoptConParallel}.
\end{proof}

\subsection{Proof of \pref{prop:FoptConPencilEig}}
\begin{proof}[Proof of \pref{prop:FoptConPencilEig}]
Note that the two operators $\sqrt{\sigma}\sqrt{E_m}$ and $U\sqrt{\rho}\sqrt{E_m}$ are parallel iff $E_m$ is supported in an eigenspace of $(\sqrt{\sigma},U\sqrt{\rho}\lsp)$ with a nonnegative eigenvalue. Therefore,  $\scrE$ is F-optimal for distinguishing $\rho$ and $\sigma$ iff each  POVM element $E_m$ is supported in an eigenspace of $(\sqrt{\sigma},U\sqrt{\rho}\lsp)$ with a nonnegative eigenvalue according to \pref{prop:FoptConParallel}. 

Next, suppose $\scrE$ is F-optimal and there exist two POVM elements supported in a same eigenspace of $(\sqrt{\sigma},U\sqrt{\rho})$. Then we can construct a nontrivial coarse graining of $\scrE$ by merging the two POVM elements, and the resulting POVM is still F-optimal according to the conclusion in the previous paragraph, so $\scrE$ is not minimal. 
In other words, if $\scrE$ is F-optimal and minimal, then no two POVM elements of $\scrE$ are supported in a same eigenspace of $(\sqrt{\sigma},U\sqrt{\rho})$.

Conversely, suppose $\scrE$ is F-optimal and no two POVM elements are supported in a same eigenspace of $(\sqrt{\sigma},U\sqrt{\rho})$. Then any nontrivial 
 coarse graining of $\scrE$ will contain a POVM element not supported in any eigenspace of $(\sqrt{\sigma},U\sqrt{\rho})$ by \lref{lem:PencilEigSpace}. Therefore, $\scrE$ is F-optimal and minimal. This observation completes the proof of \pref{prop:FoptConPencilEig}.
\end{proof}

\subsection{Proof of \pref{prop:FoptimalJM}}
\begin{proof}[Proof of \pref{prop:FoptimalJM}]

To prove the first statement in \pref{prop:FoptimalJM}, suppose, by way of contradiction, that $\scrE_1$ and $\scrE_2$ are compatible. Then they admit a common refinement, denoted by $\scrC$ henceforth. Now, the minimal  F-optimal POVMs $\scrE_1$ and $\scrE_2$ are both coarse graining of $\scrC$  and are thus equivalent according to  \pref{prop:FoptConPencilEig}. This contradiction means $\scrE_1$ and $\scrE_2$ cannot be  compatible.

Next, we turn to the second statement in \pref{prop:FoptimalJM}. Let $U$ be any unitary operator on $\caH$ that satisfies $\sqrt{\rho}\sqrt{\sigma}\lsp U = \sqrt{\sqrt{\rho}\lsp\sigma\sqrt{\rho}}$. Then each POVM element in $\scrE(p)$ or $\scrE(q)$ is supported in an eigenspace
of $(\sqrt{\sigma},U\sqrt{\rho}\lsp)$ according to  \pref{prop:FoptConPencilEig}.
Suppose, by way of contradiction, that $\scrE(p)$ and $\scrE(q)$ with $0\leq p<q\leq 1$ are equivalent.
Then $\scrE(p)$ and  $\scrE(q)$ are identical up to relabeling given that they are F-optimal and minimal and are thus automatically simple. 
Therefore, the POVM elements of $\scrE(p)$ and $\scrE(q)$ supported in the same eigenspace must coincide. Accordingly, the POVM elements of $\scrE_1$ and $\scrE_2$ supported in the same eigenspace must coincide, which means $\scrE_1$ and $\scrE_2$ are equivalent. This contradiction shows that $\scrE(p)$ and $\scrE(q)$  are not equivalent and are thus not compatible according to the first statement in \pref{prop:FoptimalJM} proved above. This observation completes the proof of  \pref{prop:FoptimalJM}. 
\end{proof}

\subsection{Proof of  \lref{lem:Parallel}}

\begin{proof}[Proof of \lref{lem:Parallel}] 
	
If  $P$ is supported in an eigenspace of $(\sqrt{\sigma},U\sqrt{\rho}\lsp)$, then $P$ is supported in either  $\Null(\rho)$ or $\supp(\rho)$ by \lref{lem:PencilEigSpacePcom} and, in the later case, 
$P$ is supported in an eigenspace of $\caM(\rho^+,\sigma)$ by \lref {lem:PencilEigGMsuppA}.

Conversely, if $P$ is supported in $\Null(\rho)$, then it is supported in the eigenspace of $(\sqrt{\sigma},U\sqrt{\rho}\lsp)$ with eigenvalue $\infty$. If  $P$ is supported in an eigenspace of $\caM(\rho^+,\sigma)$ within $\supp(\rho)$, then  $P$ is supported in an eigenspace of $(\sqrt{\sigma},U\sqrt{\rho}\lsp)$ thanks to  \lref {lem:PencilEigGMsuppA} again. This observation  completes the proof of \lref{lem:Parallel}. 
\end{proof}

\subsection{Proof of \thref{thm:FoptimalEqui}}

\begin{proof}[Proof of \thref{thm:FoptimalEqui}]
	First, we prove that Statements 1-4 in  \thref{thm:FoptimalEqui} are equivalent. If $\Pi_\rho$ and $\Pi_\sigma$ commute, then, by virtue of \lref{lem:GMsuppNull}, we can deduce that
	\begin{equation}
	\begin{gathered}
	\Null(\rho)\bot\Null(\sigma),\\
	\Null\left(\caM\left(\rho^+,\sigma\right)\right)=\Null(\rho)\oplus\Null(\sigma),\\
	\caM\left(\rho^+,\sigma\right)=\caM\left(\sigma^+,\rho\right)^+.
\end{gathered}
	\end{equation}
so  both PVMs $\scrM(\rho,\sigma)$ and $\scrM(\sigma,\rho)$ are composed of the projectors onto $\Null(\rho)$, $\Null(\sigma)$, and the eigenprojectors of $\caM\left(\rho^+,\sigma\right)$ with positive eigenvalues (any null projector is deleted by default). It follows that $\scrM(\rho,\sigma)$ and $\scrM(\sigma,\rho)$ are commuting, compatible, and equivalent, which confirms the implications $1\imply 2,3,4$. 

By construction $\Pi_\rho$ ($\Pi_\sigma$) is a sum of POVM elements in $\scrM(\rho,\sigma)$ [$\scrM(\sigma,\rho)$]. If $\scrM(\rho,\sigma)$ and $\scrM(\sigma,\rho)$ commute, then $\Pi_\rho$ and $\Pi_\sigma$ necessarily commute, which confirms the implication $2\imply 1$. If instead $\scrM(\rho,\sigma)$ and  $\scrM(\sigma,\rho)$ are equivalent, then both $\Pi_\rho$ and $\Pi_\sigma$
can be expressed as sums of POVM elements in $\scrM(\rho,\sigma)$, so 
$\Pi_\rho$ and $\Pi_\sigma$ commute as before,  which confirms the implication ${4\imply 1}$. In addition, two PVMs are compatible iff they are commuting \cite{Lahti2003coexistence,Guhne2023Compatible}, which confirms the equivalence of Statements~2 and 3. In conjunction with the above discussions we conclude that  Statements 1-4 in \thref{thm:FoptimalEqui} are equivalent.

Next, we will prove that Statement 5 is equivalent to Statements 1-4. Suppose $\scrM(\rho,\sigma)$ is the unique minimal F-optimal POVM. Then all F-optimal POVMs for distinguishing $\rho$ and $\sigma$ are refinements of $\scrM(\rho,\sigma)$, so $\scrM(\rho,\sigma)$ is a coarse graining of and equivalent to $\scrM(\sigma,\rho)$  according to \thref{thm:FoptimalReg}.   This result confirms the implication $5\imply4$, which further means  $5\imply1,2,3,4$ given that Statements 1-4  are equivalent. 
 
Conversely, suppose $\Pi_\rho$ and $\Pi_\sigma$ commute. According to \lref{lem:PencilEigGM}, there exists a unitary operator $U$ that satisfies $\sqrt{\rho}\sqrt{\sigma}\lsp U = \sqrt{\sqrt{\rho}\lsp\sigma\sqrt{\rho}}$, and 
each eigenspace of $(\sqrt{\sigma},U\sqrt{\rho}\lsp)$ coincides with either $\Null(\rho)$ or an eigenspace of $\caM\left(\rho^+,\sigma\right)$ within  $\supp(\rho)$. If $\scrE$ is an  F-optimal POVM for distinguishing $\rho$ and $\sigma$, then, according to \pref{prop:FoptConPencilEig}, each POVM element of $\scrE$  is supported in an eigenspace of $(\sqrt{\sigma},U\sqrt{\rho}\lsp)$, that is, either $\Null(\rho)$ or an eigenspace of $\caM\left(\rho^+,\sigma\right)$ within  $\supp(\rho)$. This is the case iff $\scrE$ is a refinement of $\scrM(\rho,\sigma)$. Therefore, all F-optimal POVMs are refinements of $\scrM(\rho,\sigma)$, and $\scrM(\rho,\sigma)$ is the unique minimal F-optimal POVM, which confirms the implication $1\imply5$ and completes the proof of \thref{thm:FoptimalEqui}.
\end{proof}




\section{\label{app:PureStateProof}Proofs of \psref{pro:FoptimalPure1} and \ref{pro:FoptimalPure2}}

In this section we prove \psref{pro:FoptimalPure1} and \ref{pro:FoptimalPure2}, which focus on F-optimal measurements for distinguishing a pure state and a mixed state.

\begin{proof}[Proof of \pref{pro:FoptimalPure1}]
	By assumption $\sigma$ is a pure state, $\sqrt{\sigma}=\sigma$, and $\rho\sigma\neq 0$. Let $F=\tr(\rho\sigma)$ be the fidelity between $\rho$ and $\sigma$; then $F>0$ and $\sqrt{\sqrt{\rho}\lsp\sigma\sqrt{\rho}} = \sqrt{\rho}\lsp\sigma\sqrt{\rho}/\sqrt{F}$. 
	Therefore, we can find a unitary $U$ on $\caH$ that satisfies 
	\begin{equation}
		\sqrt{\rho}\lsp\sigma  U = \sqrt{\sqrt{\rho}\lsp\sigma\sqrt{\rho}}=\frac{1}{\sqrt{F}}\sqrt{\rho}\lsp\sigma\sqrt{\rho},
	\end{equation}
	which means $U^\dag \sigma = \sqrt{\rho}\lsp\sigma/\sqrt{F}$.  In turn this result implies that
	\begin{align}
		&\Null\left(\sigma-\lambda U\sqrt{\rho}\lsp\right) = \Null\left(U^\dag \sigma-\lambda \sqrt{\rho}\lsp\right)\nonumber\\
		&=\Null\left(\sqrt{\rho}\lsp \sigma-\lambda \sqrt{F}\sqrt{\rho}\lsp\right)=\Null\left(\Pi_\rho\sigma-\lambda \sqrt{F}\lsp\Pi_\rho\right), \label{eq:FoptimalPureProof}
	\end{align}
	where the last equality holds because  $\sqrt{\rho}$ and $\Pi_\rho$ have the same null space (and support).

	Next, according to \pref{prop:FoptConPencilEig}, $\scrE$ is F-optimal iff each POVM element is supported in an eigenspace of $(\sigma,U\sqrt{\rho}\lsp)$ with a nonnegative eigenvalue, which is the case iff each POVM element is supported in an eigenspace of $(\Pi_\rho\sigma, \Pi_\rho)$ with a nonnegative eigenvalue thanks to \eref{eq:FoptimalPureProof}. This result completes the proof of \pref{pro:FoptimalPure1}.
\end{proof}

\begin{proof}[Proof of \pref{pro:FoptimalPure2}]
 \Pref{pro:FoptimalPure2}  is a simple corollary of \pref{pro:FoptimalPure1} proved above  and \lref{lem:PureEigenS} below.
\end{proof}

\begin{lemma}\label{lem:PureEigenS}
	Suppose $\dim \caH = 2$, and $\rho, \sigma$ are two pure states on $\caH$  with non-commuting density operators. 
	Let $A$ and $B$ be the representative points of $\rho$ and $\sigma$ on the Bloch sphere and let  $M$ and $N$ be their antipodal points. Then a pure state is an eigenstate of $(\rho\sigma,\rho)$ with a nonnegative eigenvalue iff its representative point lies on the major arc $\wideparen{NABM}$ that connects $M$ and~$N$.
\end{lemma}

\begin{proof}[Proof of \lref{lem:PureEigenS}]
	Without loss of generality, $\rho$ and $\sigma$ can be expressed as  $\rho=|\psi\>\<\psi|$ and $\sigma=|\varphi\>\<\varphi|$, where
	\begin{equation}\label{eq:rhosigmaPara}
		|\psi\> = |0\>,\quad|\varphi\> = \cos\frac{\theta}{2} |0\> + \sin\frac{\theta}{2}|1\>,\quad \theta\in(0,\pi);
	\end{equation}
	then $\<\psi|\varphi\>=\cos(\theta/2)>0$. Now, it is clear that any complex number $\lambda$ (including $\infty$) is a nondegenerate eigenvalue of $(\rho\sigma,\rho)$, and any pure state is an eigenstate of $(\rho\sigma,\rho)$. When $\lambda\neq\infty$, the  eigenstate associated with eigenvalue $\lambda$  spans $\Null(\rho\sigma-\lambda\rho)$ and is orthogonal to the unnormalized state $|\varphi\>\<\varphi| \psi\>-\lambda|\psi\>$. The eigenstate of $(\rho\sigma,\rho)$ corresponding to the eigenvalue $\infty$ spans  $\Null(\rho)$ and is represented by $M$ on the Bloch sphere (the antipodal point of the representative point of $\rho$).

	Next, suppose $\lambda$ is a nonnegative. As $\lambda$ increases from 0 to $\infty$, the representative point of  $|\varphi\>\<\varphi| \psi\>-\lambda|\psi\>$ (after normalization) traces the major arc $\wideparen{BMNA}$. Accordingly, the representative point of the normalized eigenstate of $(\rho\sigma,\rho)$ with eigenvalue $\lambda$ traces the major arc $\wideparen{NABM}$. This observation completes the proof of \lref{lem:PureEigenS}. 
\end{proof}

\end{document}